\documentclass[letterpaper]{article} 
\usepackage{aaai2026}  
\usepackage{times}  
\usepackage{helvet}  
\usepackage{courier}  
\usepackage[hyphens]{url}  
\usepackage{graphicx} 
\usepackage{natbib}  
\usepackage{caption} 
\usepackage{algorithm}
\usepackage{algorithmic}

\usepackage{newfloat}
\usepackage{listings}
\DeclareCaptionStyle{ruled}{labelfont=normalfont,labelsep=colon,strut=off} 
\floatstyle{ruled}
\newfloat{listing}{tb}{lst}{}
\floatname{listing}{Listing}
\nocopyright 

\usepackage{url}
\usepackage[utf8]{inputenc}
\usepackage{caption}
\usepackage{amsmath}
\usepackage{amsthm}
\usepackage{booktabs}
\usepackage[shortlabels]{enumitem}

\usepackage{natbib}
\usepackage{cleveref}

\usepackage{amssymb}
\usepackage{amsfonts}
\usepackage{xcolor}

\newcommand{\bfx}{\mathbf{x}}
\newcommand{\bfy}{\mathbf{y}}
\newcommand{\bfe}{\mathbf{e}} 
\newcommand{\ttT}{\mathsf{T}}
\newcommand{\half}{\mathrm{half}}
\newcommand{\bbR}{\mathbb{R}}
\newcommand{\argmax}{\mbox{argmax}}
\newcommand{\cD}{\mathcal{D}} 
\newcommand{\cP}{\mathcal{P}} 
\newcommand{\cY}{\mathcal{Y}} 
\newcommand{\iid}{\mathrm{i.i.d.}} 
\newcommand{\N}{\mathbb{N}}
\newcommand{\var}{\mathrm{Var}}
\newcommand{\R}{\mathbb{R}}
\newcommand{\E}{\mathbb{E}}
\newcommand{\bone}{\mathbf{1}}
\newcommand{\maxprob}{\mathfrak{d}}
\newcommand{\tX}{\tilde{X}}

\newcommand{\multi}[2]{\mathrm{Mult}\left(#1, #2\right)}
\newcommand{\bino}[2]{\mathrm{Bin}\left(#1, #2\right)}
\newcommand{\ber}[1]{\mathrm{Ber}\left(#1\right)}
\newcommand{\bp}{\mathbf{p}}
\newcommand{\bx}{\mathbf{x}}
\newcommand{\bv}{\mathbf{v}}
\newcommand{\Z}{\mathbb{Z}}
\newcommand{\tr}{\tilde{r}}
\newcommand{\kl}[2]{\mathrm{D_{KL}}\left(#1\middle\|#2\right)}
\newcommand{\tv}[2]{\mathrm{D_{TV}}\left(#1, #2\right)}

\newcommand{\cE}{\mathcal{E}}
\newcommand{\bC}{\mathbf{C}}

\newcommand{\CEF}[1]{\bC^{\mathrm{EF},{#1}}}
\newcommand{\imaxa}[1]{i^*_a(#1)}
\newcommand{\cov}[1]{\mathrm{Cov}[#1]}
\newcommand{\cN}{\mathcal{N}}
\newcommand{\diag}{\mathrm{diag}}
\newcommand{\supp}{\mathrm{supp}}

\newcommand{\bzero}{\mathbf{0}}
\newcommand{\tSigma}{\tilde{\Sigma}}
\newcommand{\be}{\mathbf{e}}
\newcommand{\cA}{\mathcal{A}}
\newcommand{\cAbal}{\mathcal{A}_{\mathrm{bal}}}
\newcommand{\bbP}{\mathbb{P}}
\newcommand{\bu}{\mathbf{u}}

\newcommand{\of}{\overline{f}}

\newtheorem{theorem}{Theorem}
\newtheorem{lemma}{Lemma}

\newtheorem{observation}{Observation}
\newtheorem{definition}{Definition}

\title{Fair Allocation of Indivisible Goods with Variable Groups}

\author{
    Paul G\"{o}lz\textsuperscript{\rm 1},
    Ayumi Igarashi\textsuperscript{\rm 2},
    Pasin Manurangsi\textsuperscript{\rm 3},
    Warut Suksompong\textsuperscript{\rm 4}
}
\affiliations{
    \textsuperscript{\rm 1}Cornell University, USA\\
    \textsuperscript{\rm 2}University of Tokyo, Japan\\
    \textsuperscript{\rm 3}Google Research, Thailand\\
    \textsuperscript{\rm 4}National University of Singapore, Singapore
}

\allowdisplaybreaks

\begin{document}

\maketitle

\begin{abstract}
We study the fair allocation of indivisible goods with \emph{variable groups}.
In this model, the goal is to partition the agents into groups of given sizes and allocate the goods to the groups in a fair manner.
We show that for any number of groups and corresponding sizes, there always exists an \emph{envy-free up to one good (EF1)} outcome, thereby generalizing an important result from the individual setting.
Our result holds for arbitrary monotonic utilities and comes with an efficient algorithm.
We also prove that an EF1 outcome is guaranteed to exist even when the goods lie on a path and each group must receive a connected bundle.
In addition, we consider a probabilistic model where the utilities are additive and drawn randomly from a distribution.
We show that if there are $n$ agents, the number of goods~$m$ is divisible by the number of groups~$k$, and all groups have the same size, then an \emph{envy-free} outcome exists with high probability if $m = \omega(\log n)$, and this bound is tight.
On the other hand, if $m$ is not divisible by~$k$, then an envy-free outcome is unlikely to exist as long as $m = o(\sqrt{n})$.
\end{abstract}

\section{Introduction}

Fairly allocating limited resources is a fundamental societal challenge, with applications ranging from dividing household supplies among families to distributing personnel among schools or other public institutions.
The problem of \emph{fair division} has accordingly received interest in several disciplines, including in computational social choice and multi-agent systems \citep{BouveretChMa16,Markakis17,Aziz20,Walsh20}. 

Most of the work in fair division assumes that each recipient of a bundle of resources is an \emph{individual} agent, represented by a single preference.
However, when distributing resources among families, schools, or institutions, each recipient in fact consists of multiple agents.
Although these agents share the same set of resources and derive full value from the resources in their set, they may have differing preferences over the resources.
This has motivated several researchers to study the \emph{(fixed-)group} model, where the agents are partitioned into groups and the aim is to allocate the resources in a fair manner among the groups \citep{ManurangsiSu17,ManurangsiSu24,GhodsiLaMo18,SegalhaleviNi19,KyropoulouSuVo20,SegalhaleviSu23,CaragiannisLaSh25,GolzYa26,ManurangsiMe25}.
For instance, \citet{ManurangsiSu22} investigated the notion of \emph{envy-freeness up to $c$ goods (EF$c$)}.
This means that no agent would prefer another group's bundle over her own group's if some set of at most $c$ goods were removed from the other group's bundle.
Manurangsi and Suksompong proved that when the number of groups is constant, there exists an EF$c$ allocation for $c = O(\sqrt{n})$, where $n$ denotes the total number of agents, and this bound is also asymptotically tight.

The fixed-group model is appropriate when membership in the groups is predetermined, as in the allocation among families or countries.
In other applications, however, the resource allocator can select the partition of agents into groups alongside the allocation.
This is the case, for example, when dividing workers in an organization into teams and assigning resources to these teams.
In light of this, \citet[Sec.~5]{KyropoulouSuVo20} proposed a \emph{variable-group} model, in which a partition of the agents into groups can be chosen along with an allocation of the resources.\footnote{For further motivation of the variable-group model, we refer to the papers by \citet{KyropoulouSuVo20} and \citet{SegalhaleviSu21}.}

When the resources consist of \emph{divisible} items such as time, \citet{SegalhaleviSu21} proved that an envy-free outcome always exists in the variable-group model.
On the other hand, if the resources contain \emph{indivisible} items such as books or gym equipment, \citet{KyropoulouSuVo20} showed that an EF1 outcome can be ensured in the case of two groups, and an outcome satisfying a relaxation of proportionality---which is fundamentally weaker than EF1---can be satisfied for any number of groups.
This differs from the fixed-group model, where even for two groups, the optimal EF$c$ guarantee deteriorates as the number of agents grows.

In this paper, we expand and deepen our understanding of fairness, particularly envy-freeness, when allocating indivisible goods in the variable-group model.
For instance, we study the following question: does an EF1 outcome exist for any number of groups and any corresponding sizes, or is it sometimes necessary to relax the notion to EF$c$ for some (possibly non-constant) $c$?
As we shall see, the flexibility provided by this model enables remarkably strong fairness guarantees to be made.

\subsection{Our Results}

In \Cref{sec:EF1-existence}, we answer the question above in the positive: for any desired group sizes, an EF1 outcome exists and, moreover, can be computed efficiently (\Cref{thm:ef1-existence}).
This result holds even for arbitrary monotonic utilities, and significantly generalizes the well-known result by \citet{LiptonMaMo04} that EF1 allocations always exist in the \emph{individual setting} where each group has size one.
Our algorithm is a careful extension of Lipton et al.'s classic \emph{envy cycle elimination algorithm}.
This positive result stands in contrast to the fixed-group setting, where $c = \Omega(\sqrt{n})$ is required to guarantee the existence of an EF$c$ allocation \citep{ManurangsiSu22}.

In \Cref{sec:EF1-connected}, we strengthen the previous result by showing that, if the goods lie on a path, EF1 allocations exist even when each group must receive a connected bundle (\Cref{thm:EF1:group}).
The connectivity requirement is natural when the goods correspond to time slots or offices along a corridor, and has been studied in several papers \citep{BouveretCeEl17,BeiIgLu22,BiloCaFl22,Igarashi23}.
This existence result also holds for monotonic utilities, and generalizes results from the individual setting by \citet{BiloCaFl22} and \citet{Igarashi23}.
However, like in the individual setting, the result does not come with an efficient algorithm.

Finally, in \Cref{sec:EF-asymptotic}, we consider a probabilistic model where the utilities are additive and each agent's utility for each good is drawn independently at random from a non-atomic distribution.\footnote{A distribution is called \emph{non-atomic} if it does not put a positive probability on any single point.}
We are interested in when an \emph{envy-free}\footnote{Envy-freeness corresponds to EF$c$ for $c = 0$. 
An envy-free outcome does not always exist, e.g., when there are two groups and only
one valuable good.} outcome exists with high probability (that is, with probability approaching~$1$ as the number of agents~$n$ grows), assuming that the number of groups~$k$ is fixed and all groups have the same size.
Interestingly, we show that the existence depends on whether the number of goods~$m$ is divisible by~$k$.
On the one hand, if $m$ is divisible by $k$, the transition occurs at merely $\Theta(\log n)$: an envy-free outcome is unlikely to exist if $m = o(\log n)$ (\Cref{thm:lb-div-asympt}), but likely to exist if $m = \omega(\log n)$ (\Cref{thm:existence-div}).
On the other hand, if $m$ is not divisible by $k$, such an outcome is unlikely to exist as long as $m = o(\sqrt{n})$ (\Cref{thm:lb-indiv-asympt}).\footnote{In contrast, a result in the fixed-group setting by \citet{ManurangsiSu17} implies that an envy-free outcome is likely to exist when $m = \Omega(n\log n)$.}

\subsection{Further Related Work}
\label{sec:related-work}

As mentioned earlier, a number of authors have investigated fair division among groups, mostly focusing on the fixed-group model.
Besides envy-freeness, \citet{ManurangsiSu22} also obtained bounds on \emph{proportionality} as well as \emph{consensus $1/k$-division}---the latter is even more stringent than envy-freeness, as it requires agents to value all bundles of goods equally.
While these bounds were already tight in terms of $n$,
\citet{CaragiannisLaSh25} and \citet{ManurangsiMe25} recently improved their dependence on $k$.
\citet{BuLiLi23}, \citet{BarmanEbLa25}, as well as \citet{KawaseRoSa25} studied settings that can be interpreted as special cases of the fixed-group model.
For example, Bu et al.'s setting corresponds to the fixed-group model when each group has size two.

The probabilistic model we consider in \Cref{sec:EF-asymptotic} falls under the framework of \emph{asymptotic fair division}.
This framework was introduced by \citet{DickersonGoKa14} and subsequently studied in several papers \citep{KurokawaPrWa16,Suksompong16,ManurangsiSu20,ManurangsiSu21,ManurangsiSu25,BaiGo22,BaiFeGo22,BenadeHaPs24,ManurangsiSuYo25,YokoyamaIg25}.
The motivation is that, since an allocation satisfying envy-freeness (or some other fairness notion) does not always exist, it is natural to ask \emph{when} an envy-free allocation is likely to exist if the utilities are drawn at random.
We highlight two relations between our results and known results from this line of work.
Firstly, \citet{ManurangsiSu17} examined the fixed-group model and showed that an envy-free allocation is unlikely to exist unless $m = \Omega(n)$.
This contrasts with our results, which show that existence is already likely in the variable-group model when $m = \omega(\log n)$.
Secondly, in the individual setting, \citet{ManurangsiSu20,ManurangsiSu21} proved that the threshold for the existence of envy-free allocations differs according to whether $m$ is divisible by~$n$.
However, the (multiplicative) gap in their case is only logarithmic (i.e., $\Theta(n)$ vs.~$\Theta(n\log n/\log\log n)$), whereas our gap is much larger (i.e., $\Theta(\log n)$ vs.~$\Omega(\sqrt{n})$).

\section{Preliminaries}

For any positive integer $t$, let $[t] := \{1,2,\dots,t\}$.
Let $k \ge 2$ and $n_1,\dots,n_k$ be positive integers, $N = [n]$ be a set of $n = n_1 + \dots + n_k$ agents, and $M = [m]$ be a set of $m$ goods.
A \emph{bundle} refers to a (possibly empty) set of goods.
Each agent $a\in N$ has a utility function $u_a$ over the sets of goods in $M$; for a single good $g\in M$, we sometimes write $u_a(g)$ instead of $u_a(\{g\})$.
We assume that the utilities are \emph{monotonic}, meaning that $u_a(M')\le u_a(M'')$ for any $a\in N$ and $M'\subseteq M''\subseteq M$, and \emph{normalized}, i.e., $u_a(\emptyset) = 0$.
The utilities are called \emph{additive} if $u_a(M') = \sum_{g\in M'} u_a(g)$.
When utilities are non-additive, we assume that an algorithm can query any agent's utility for any set of goods in constant time.
An \emph{instance} in the variable-group model consists of the set of agents $N$, the set of goods $M$, the agents' utility functions, and the desired group sizes $n_1,\dots,n_k$.

We would like to partition the $n$ agents into $k$ groups $C_1,\dots,C_k$ of sizes $n_1,\dots,n_k$, respectively, and allocate the $m$ goods among these groups.
We write $C = (C_1,\dots,C_k)$.
An \emph{allocation} $A = (A_1,\dots,A_k)$ consists of $k$ disjoint bundles, where bundle~$A_i$ is allocated to the $i$-th group; it is called \emph{complete} if $A_1\cup \dots\cup A_k = M$.
An \emph{outcome} consists of a partition of agents $C$ along with a complete allocation of goods $A$.
An outcome $(C, A)$ is
\begin{itemize}
\item \emph{envy-free up to $c$ goods (EF$c$)}, for a given non-negative integer~$c$, if for every $i,j\in [k]$ and every agent $a\in C_i$, there exists a set $B\subseteq A_j$ with $|B|\le c$ such that $u_a(A_i) \ge u_a(A_j\setminus B)$;
\item \emph{envy-free} if it is EF$0$.
\end{itemize}
The \emph{individual setting} refers to a special case of this model where $n_1 = \dots = n_k = 1$ (and therefore $k = n$).

\section{EF1: Existence and Computation}
\label{sec:EF1-existence}

Recall that in the individual setting, a classic result of \citet{LiptonMaMo04} states that for any instance with arbitrary monotonic utilities, an EF1 allocation exists.
Moreover, such an allocation can be found in time $O(mn^3)$ via the \emph{envy cycle elimination algorithm}.
Intuitively, this algorithm allocates one good at a time in an arbitrary order.
The algorithm maintains an ``envy graph'', which is a directed graph that captures the envy relations among the agents.
In particular, the vertices represent the agents, and there is an edge from one agent to another agent if and only if the former agent envies the latter agent.
Hence, an agent is unenvied exactly when the corresponding vertex has no incoming edge. 
Each good is assigned to an unenvied agent, and if the assignment causes a cycle to be formed in the envy graph, the cycle is eliminated by giving each agent's bundle to the preceding agent on the cycle.
Once all cycles have been eliminated, the next good can be assigned to an unenvied agent, and EF1 is maintained throughout the algorithm.

A priori, it may appear that the envy cycle elimination algorithm is not well-suited for the group setting.
Indeed, it is unclear when a group should be considered to ``envy'' another group, as different agents in the same group may have differing opinions about the groups' bundles.
Nevertheless, we show that adopting an alternative perspective on the algorithm allows for its generalization to the group setting.
Specifically, instead of moving the \emph{bundles} as in the typical interpretation, we can reinterpret the envy cycle elimination step as moving the \emph{agents} along the cycles instead.
This interpretation enables a generalization of the algorithm to accommodate groups, as it permits the reassignment of individual agents according to their envy relations while preserving group sizes, thereby leveraging the flexibility of the variable-group model.
In addition to extending the seminal result of \citet{LiptonMaMo04}, the following theorem also strengthens a variable-group guarantee due to \citet[Thm. 5.6]{KyropoulouSuVo20}, which holds under additive utilities for a much weaker notion than EF1. 

\begin{theorem} \label{thm:ef1-existence}
For any instance with arbitrary monotonic utilities, there exists an EF1 outcome.
Moreover, such an outcome can be computed in time $O(mn^3)$.
\end{theorem}

\begin{proof}
We use the following generalization of the envy cycle elimination algorithm.
\begin{enumerate}
\item Let $C$ be an arbitrary partition of the $n$ agents into groups of sizes $n_1,\dots,n_k$, and let $A$ be an empty allocation.
\item Construct an \emph{envy graph}, which is a directed graph with the $k$ groups as the vertices; this graph will be updated as the algorithm proceeds.
For each agent, add an edge from the agent's group to another group if the agent envies the latter group.
(Thus, the graph initially contains no edges.)
In particular, there can be multiple edges from one vertex to another vertex.
\item Take an arbitrary unallocated good, and allocate it to any group with no incoming edge in the envy graph.
Update the envy graph.
\item If there is at least one directed cycle in the envy graph, consider an arbitrary cycle.
For each edge in the cycle, move the agent associated with this edge to the group that the agent envies.
Update the envy graph.
If there is still a directed cycle in the envy graph, repeat this step.
\item If there is still an unallocated good, go back to Step~3.
Otherwise, output the current outcome $(C,A)$.
\end{enumerate}

We first show that the algorithm outputs a valid outcome.
To this end, we prove that each time we eliminate a cycle in Step~4, the number of edges in the envy graph decreases.
For each agent \emph{not} associated with an edge in the cycle, the agent's own bundle as well as all other bundles remain the same, so the number of envy edges from the agent also remains the same.
On the other hand, for each agent associated with an edge in the cycle, the agent is assigned to a better bundle in her view among the $k$ bundles, so the number of envy edges from the agent decreases by at least one.
Hence, the total number of envy edges decreases, which means that the process of eliminating cycles must terminate.
When the envy graph contains no cycle, there must exist a vertex with no incoming edge, and we can allocate the next good to the corresponding group.
It follows that all goods are allocated.
Moreover, the sizes of the $k$ groups remain $n_1,\dots,n_k$ throughout, so the algorithm outputs a valid outcome.

Next, we show that the outcome returned by the algorithm is EF1.
Specifically, we prove that at every point during the execution of the algorithm, the partition $C$ and the (possibly incomplete) allocation $A$ together yield EF1.
This is true at the beginning of the algorithm, as the allocation is empty.
When a good is allocated, it is allocated to a group with no incoming edge, so any envy towards the group can be eliminated by removing this good.
Moreover, when a cycle is eliminated, for each agent \emph{not} associated with an edge in the cycle, the agent's own bundle as well as all other bundles remain the same, so the EF1 invariant is maintained.
On the other hand, for each agent associated with an edge in the cycle, the agent is assigned to a better bundle in her view among the $k$ bundles.
Since any envy that the agent has towards another bundle can be eliminated by removing a good from the bundle before, the same remains true afterwards, and the EF1 invariant is again maintained.

Finally, we analyze the running time of the algorithm.
There are $m$ allocated goods, and each allocated good increases the number of envy edges by at most $n$.
Finding and eliminating a directed cycle can be done in time $O(k^2)$, and the elimination decreases the number of envy edges by at least one.
As the algorithm can query any agent's utility for any set of goods in constant time, updating the envy graph takes time $O(nk)$.
Since $k\le n$, the algorithm runs in time $O(mn^3)$.
\end{proof}

\section{EF1: Adding Connectivity Constraints}
\label{sec:EF1-connected}

Having established the general existence of EF1 outcomes, in this section, we impose an additional requirement in the form of connectivity.
Specifically, we assume that the goods lie on a path $1,2,\dots,m$.
An allocation $A$ is said to be {\em connected} if $A_j$ forms an interval on the path for each $j\in [k]$. 
An outcome $(C,A)$ is called {\em connected} if $A$ is connected.

In the individual setting, \citet{Igarashi23} proved that a connected EF1 allocation is guaranteed to exist, thereby strengthening an earlier result of \citet{BiloCaFl22}, which holds for EF2.
For their proofs, Bil\`{o} et al.~and Igarashi developed a discretization approach using \emph{Sperner's lemma} \citep{Sperner28}. 
This idea was originally due to \citet{Su99}, who provided an elegant proof that an envy-free allocation of a cake (i.e., a heterogeneous divisible good) always exists. 
Su's proof encodes possible divisions as points in the standard simplex and employs a triangulation of the simplex along with a coloring of each vertex of the triangulation. By applying Sperner's lemma, Su showed the existence of an elementary simplex (i.e., a simplex not composed of smaller simplices in the triangulation) whose vertices correspond to allocations that are similar to one another, where each allocation ensures that a different agent is envy-free. 
An infinite sequence of such simplices, with their sizes progressively shrinking, converges to an envy-free allocation.

Our main result of this section generalizes the results of \citet{BiloCaFl22} and \citet{Igarashi23} to the group setting.
The result is formally stated as follows.

\begin{theorem}\label{thm:EF1:group}
For any instance with arbitrary monotonic utilities such that the goods lie on a path, there exists a connected EF1 outcome.
\end{theorem}

Note that unlike \Cref{thm:ef1-existence}, this theorem does not come with efficient computation.
Indeed, the question of whether a connected EF1 allocation can be computed efficiently is open even in the individual setting \citep{Igarashi23}.

We start by recalling some basic notions of combinatorial topology. 
For any positive integer~$k$, a \emph{$(k-1)$-simplex}~$S$ is the convex hull of $k$ \emph{main vertices} $\bfx_1,\bfx_2,\ldots,\bfx_k$. 
We use the notation $S=\langle \bfx_1,\bfx_2,\ldots,\bfx_k \rangle$. 
The \emph{$(k-1)$-standard simplex} $\Delta^{k-1}$ is the $(k-1)$-simplex whose main vertices are given by the $j$-th unit vectors $\bfe^{j} \in \{0,1\}^k$ for $j \in [k]$, where $e^j_h=1$ if $h=j$ and $e^j_h=0$ otherwise. 

A {\em triangulation} $\ttT$ of a $(k-1)$-simplex $S$ is a collection of smaller ($k-1$)-simplices $S_1,S_2,\ldots,S_h$ such that the union of simplices $S_j$ for $j \in [h]$ is $S$, and for each $i \neq j$, the intersection $S_i \cap S_j$ is either empty or a face common to $S_i$ and $S_j$. 
We call $S_1,S_2,\ldots,S_h$ {\em elementary simplices}, and write $V(\ttT)$ for the set of vertices of a triangulation $\ttT$. 

For a triangulation $\ttT$ of a $(k-1)$-simplex $S$, a {\em coloring} is a mapping $\lambda \colon V(\ttT) \rightarrow [k]$ that assigns a color in $[k]$ to each vertex $\bfx \in V(\ttT)$. 
A coloring $\lambda$ is called {\em proper} if we can write $S=\langle \bfx_1,\bfx_2,\ldots,\bfx_k \rangle$ in such a way that if a vertex $\bfx \in V(\ttT)$ is colored with color~$j$ (i.e., $j =\lambda (\bfx)$), then $\bfx_j$ is a vertex of the minimal face containing $\bfx$. 
For example, if $k = 3$, then a $(k-1)$-simplex can be viewed as a triangle.
In a proper coloring of the simplex, the three main vertices $\bfx_1,\bfx_2,\bfx_3$ are colored $1,2,3$, respectively, and each vertex $\bfx$ on an edge between $\bfx_i$ and $\bfx_j$ is colored either $i$ or $j$.

In our proof, the space of connected allocations is encoded by the positions of $k-1$ ``knives''. 
Like \citet{BiloCaFl22} and \citet{Igarashi23}, we consider a triangulation of this space, where each knife is at either a vertex or an edge of the path.
More precisely, consider the following simplex: 
\[
S_m := \left\{ \bfx \in \bbR^{k-1} \,\middle|\, \frac{1}{2} \le x_1 \le \dots \le x_{k-1} \le m +\frac{1}{2}  \right\}.
\]
Let $\ttT_{\half}$ be a {\em Kuhn's triangulation} of $S_m$ \citep{DengQiSa12,BiloCaFl22,Igarashi23}. 
The vertices of this triangulation are given by 
\[
V(\ttT_{\half})= \left\{  \bfx \in S_m  \,\middle|\, x_i  \in  \left\{\frac{1}{2},1, \ldots, m +\frac{1}{2} \right\}, \forall i \right\},
\]
where each elementary simplex $S=\langle \bfx_1,\bfx_2,\ldots,\bfx_k \rangle$ of $\ttT_{\half}$ satisfies the property that there exists a permutation $\phi \colon [k] \rightarrow [k]$ such that 
$
\bfx_{\phi(i+1)}=\bfx_{\phi(i)}+\frac{1}{2} \bfe^{\phi(i)}~\mbox{for each}~i \in [k-1] 
$.
Each vertex $\bfx \in V(\ttT_{\half})$ yields a partial allocation $A({\bfx})=(A_1(\bfx),A_2(\bfx),\ldots,A_k(\bfx))$ such that for each $j \in [k]$, the $j$-th bundle is given by
\[
A_j(\bfx) = \,\{ y \in [m]  \mid x_{j-1} < y < x_j \}, 
\]
with $x_0 = \frac12$ and $x_k = m+\frac12$.

With appropriate coloring and rounding, \citet{Igarashi23} demonstrated that a desired elementary simplex, guaranteed by Sperner's lemma, can be rounded to produce a connected EF1 allocation. 
The proof relies on the notion of a {\em virtual utility}, $\hat{u}_a(\bfx, j)$, defined for each vertex $\bfx$ of the triangulation and each index $j \in [k]$. 
This virtual utility determines agent~$a$'s most preferred bundle in a partial allocation corresponding to the vertex $\bfx$; in particular, $\hat{u}_a(\bfx, j) = 0$ if $A_j(\bfx) = \emptyset$, and $\hat{u}_a(\bfx, j) \ge 0$ otherwise.
\citet[Alg.~1]{Igarashi23} presented an algorithm that, given an elementary simplex $S=\langle \bfx_1,\bfx_2,\ldots,\bfx_k \rangle$ of $\ttT_{\half}$, produces an allocation $A=(A_1,A_2,\ldots,A_k)$ with the following property: each agent's estimate of the $j$-th bundle based on the virtual utility is upper-bounded by her true utility of $A_j$ and lower-bounded by the her ``up-to-one utility'' of $A_j$, as defined next.
For any connected subset of goods $I$, the \emph{up-to-one utility} $u^-_a(I)$ of agent $a$ is defined as
\begin{align*}
	u^-_a(I) :=
	\begin{cases}
		0 & \text{if $I = \emptyset$}; \\
		\min \big\{u_a(I\setminus\{g\}) \mid g\in I \\
		\text{ such that } I\setminus\{g\}\text{ is connected} \big\} & \text{if $I \neq \emptyset$.}
	\end{cases}
\end{align*}

The following lemma is stated as Lemma~3.2 in the work of \citet{Igarashi23}.\footnote{Although Igarashi's Lemma~3.2 is stated for the case $k = n$, the same proof also works when $k\ne n$.}

\begin{lemma}[\citealt{Igarashi23}]\label{lem:ef1:approx-correct}
	Consider the triangulation $\ttT_{\half}$ of $S_m$. There exists an algorithm that, given any elementary simplex $S=\langle \bfx_1,\bfx_2,\ldots,\bfx_k \rangle$ of $\ttT_{\half}$, returns a connected allocation $(A_1,A_2,\ldots,A_k)$ such that for every agent $a \in N$ and every pair of indices $i, j\in [k]$, we have
	$u_a(A_j) \ge \hat u_a(\bfx_i, j) \ge u_a^-(A_j)$.
\end{lemma}

Our main lemma of this section is as follows. 

\begin{lemma} \label{lem:main-coloring-to-assignment}
Let $\lambda_1, \dots, \lambda_n: V(\ttT_\half) \to [k]$ be any proper colorings. Then, there exist an elementary simplex $S^* =\langle \bfx^*_1,\bfx^*_2,\ldots,\bfx^*_k \rangle$ of $\ttT_{\half}$ and $\pi: N \to [k]$ such that
\begin{enumerate}
\item $|\pi^{-1}(i)| = n_i$ for each $i \in [k]$, and
\item $\pi(a) \in \{\lambda_a(\bfx^*_h) \mid h \in [k]\}$ for each $a \in N$.
\end{enumerate}
\end{lemma}

To establish \Cref{lem:main-coloring-to-assignment}, we will use the following result shown by \citet{IgarashiMe24} as Lemma 4 in their work. We state a simplified version of their lemma below.\footnote{We only state the case where the numbers $n_i$ are integers, but their original lemma handles the non-integer case too.}

\begin{lemma}[\citealt{IgarashiMe24}] \label{lem:assignment}
Consider a bipartite graph $H = (N, [k]; E)$ together with non-negative edge weights $(w_e)_{e \in E}$. 
Let $\delta_H(v)$ denote the set of edges incident to vertex~$v$ in $H$.
Suppose that the following holds:
\begin{align*}
\sum_{e \in \delta_H(a)} w_e &= 1 & &\forall a \in N; \\
\sum_{e \in \delta_H(i)} w_e &= n_i & &\forall i \in [k].
\end{align*}
Then, there exists an assignment $\pi: N \to [k]$ such that
\begin{enumerate}
\item $|\pi^{-1}(i)| = n_i$ for each $i \in [k]$, and
\item $(a, \pi(a)) \in E$ for each $a \in N$.
\end{enumerate}
\end{lemma}

\begin{proof}[Proof of \Cref{lem:main-coloring-to-assignment}]
For each $a \in N$ and $\bfx \in V(\ttT_{\half})$, we define $f^{(a)}(\bfx) = \bfe^{\lambda_a(\bfx)}$. 
We then extend it in an affine manner to $\of^{(a)}: S_m \to \Delta^{k-1}$ as follows. 
For each $\bfx \in S_m$, let $S = \left<\bfx_1, \dots, \bfx_k\right>$ be any elementary simplex of $\ttT_{\half}$ containing it. Then, write $\bfx = \sum_{h \in [k]} \alpha_h \bfx_h$, and let
\begin{align*}
\of^{(a)}(\bfx) = \sum_{h\in[k]} \alpha_h f^{(a)}(\bfx_h).
\end{align*}
The affine extension on each simplex yields a continuous function on the entire~$S_m$ since adjacent simplices share faces and vertex values agree on those faces.
Finally, we define $\of: S_m \to \Delta^{k-1}$ by letting
\begin{align*}
\of(\bfx) = \frac{1}{n} \sum_{a \in N} \of^{(a)}(\bfx) & &\forall \bfx \in S_m.
\end{align*}
It can be shown that $\of$ is surjective. 
We defer the full proof of this fact to \Cref{app:surjective} (see also, e.g., \citealt{MeunierSu19}). 
Since $\of$ is surjective, there exists $\bfx^* \in S_m$ such that 
\begin{align*}
\of(\bfx^*)_i = \frac{n_i}{n} & &\forall i \in [k].
\end{align*}
Let $S^* =\langle \bfx^*_1,\bfx^*_2,\ldots,\bfx^*_k \rangle$ be any elementary simplex of $\ttT_{\half}$ containing $\bfx^*$. 
We define the graph $H = (N, [k]; E)$ as follows:
\begin{itemize}
\item $(a, i) \in E$ if and only if $\of^{(a)}(\bfx^*)_i > 0$, and
\item $w_{(a, i)} = \of^{(a)}(\bfx^*)_i$.
\end{itemize}
By definition, we have 
\begin{align*}
\sum_{e \in \delta_H(a)} w_e = \sum_{i \in [k]} \of^{(a)}(\bfx^*)_i = 1 & & \forall a \in N
\end{align*} and
\begin{align*}
\sum_{e \in \delta_H(i)} w_e = \sum_{a \in N} \of^{(a)}(\bfx^*)_i = n \cdot \of(\bfx^*) = n_i & &\forall i \in [k].
\end{align*}
Thus, we can apply \Cref{lem:assignment} to obtain $\pi: N \to [k]$ satisfying the two conditions of the lemma. 
The first condition of \Cref{lem:assignment} ensures the first condition of \Cref{lem:main-coloring-to-assignment}. 
Furthermore, the second condition of \Cref{lem:assignment} states that $(a, \pi(a)) \in E$ for every $a \in N$. 
By definition of $H$, this implies that $\of^{(a)}(\bfx^*)_{\pi(a)} > 0$. 
By definition of $\of^{(a)}$, this in turn ensures that $\pi(a) = \lambda_a(\bfx^*_h)$ for some $h \in [k]$. 
Hence, the second condition of \Cref{lem:main-coloring-to-assignment} is also satisfied. 
\end{proof}

We now show how to use \Cref{lem:main-coloring-to-assignment} to prove \Cref{thm:EF1:group}. 

\begin{proof}[Proof of Theorem \ref{thm:EF1:group}]
Based on the virtual utilities, define a coloring ${\lambda}_a \colon V(\ttT_\half) \to [k]$ for each $a \in N$ such that
\[ 
{\lambda}_a(\bfx) \in \argmax \{ {\hat u}_a(\bfx, j) \mid j \in [k]~\mbox{such that}~A_j(\bfx) \neq \emptyset \}. 
\]
These colorings were shown to be proper by \citet{Igarashi23}.

\Cref{lem:main-coloring-to-assignment} then yields an elementary simplex $S^* = \langle \bfx^*_1,\bfx^*_2,\ldots,\bfx^*_k \rangle$ of $\ttT_{\half}$ and $\pi: N \to [k]$ that satisfy the two conditions of the lemma. 
Let $C_i = \pi^{-1}(i)$ for each $i \in [k]$; the first condition of the lemma ensures that $|C_i| = n_i$.
Next, we construct the allocation $A^* = (A^*_1, \dots, A^*_k)$ by invoking \Cref{lem:ef1:approx-correct} on $S^*$. 
We claim that $(C, A^*)$ is an EF1 outcome. 
To see this, consider any agent $a \in N$ (belonging to $C_{\pi(a)}$) and any $j \in [k]$. 
By the second condition of \Cref{lem:main-coloring-to-assignment}, $\pi(a) = \lambda_a(\bfx^*_h)$ for some $h \in [k]$. 
We thus have
\begin{alignat*}{2}
u_a(A^*_{\pi(a)})
&\ge \hat u_a (\bfx^*_h, \pi(a)) &&\text{(by Lemma~\ref{lem:ef1:approx-correct})}
\\
&\ge \hat u_a (\bfx^*_h, j) &&\text{(since $\pi(a) = \smash{\lambda_{a}}(\bfx^*_{h})$)}
\\
&\ge u_a^-(\smash{A^*_{j}}) \qquad \quad &&\text{(by Lemma~\ref{lem:ef1:approx-correct})}.
\end{alignat*}
Hence, $(C, A^*)$ is an EF1 outcome, as desired.
\end{proof}

\section{Envy-Freeness: Asymptotic Existence}
\label{sec:EF-asymptotic}

In this section, we turn our attention to the asymptotic existence of envy-free outcomes.
Specifically, we consider a setting where the utilities are additive and, for each $a\in N$ and $g\in M$, the utility $u_a(g)$ is drawn independently from a given distribution $\cD$.
This distribution is assumed to be non-atomic (i.e., does not put a positive probability on any single point) 
and has a support contained in the interval\footnote{If the support is contained in $[0,t]$ for some $t > 1$, we can scale down all utilities by $t$.} $[0, 1]$. 
We also assume throughout this section that all groups have the same size, that is, $n_1 = \cdots = n_k = n/k$---in particular, $n$ is divisible by~$k$.
We fix $k$ and investigate the asymptotic (non-)existence of EF outcomes as $n$ grows.
We say that an event occurs \emph{with high probability} if the probability that it occurs approaches~$1$ as $n\rightarrow\infty$.

\subsection{Divisible Case}

We first consider the case where the number of goods~$m$ is divisible by~$k$. 
Intuitively, this is an ``easier'' case, as it is possible to give every group the same number of goods.
For this case, we show that the phase transition occurs at $m = \Theta(\log n)$, as stated in the following two theorems.

\begin{theorem} \label{thm:lb-div-asympt}
For any fixed $k \geq 2$, if $m \leq \frac{\ln n}{4}$, then with high probability, no envy-free outcome exists.
\end{theorem}

\begin{theorem} \label{thm:existence-div}
For any fixed $k \geq 2$, utility distribution $\cD$, and $\beta \in (0,1)$, there exists a constant $C_{k, \beta, \cD}$ such
that, for any sufficiently large $n$ and any $m \geq C_{k, \beta, \cD} \cdot \ln n$ with the property that $m$ is divisible by $k$,  
there is a polynomial-time algorithm that computes an envy-free outcome with probability at least $1 - \beta$.
\end{theorem}

To facilitate the proofs, we introduce some definitions.
\begin{itemize}
\item Let $\cA$ denote the set of all allocations.
\item An allocation $A = (A_1, \dots, A_k) \in \cA$ is called \emph{balanced} if $|A_1| = \cdots = |A_k| = m/k$. 
Let $\cAbal$ be the set of all balanced allocations.
\item For $A \in \cA$ and $a \in N$, let $\imaxa{A} = \argmax_{i \in [k]} u_a(A_i)$, ties broken arbitrarily. 
For each $i \in [k]$, let $p^{A}_i = \Pr[\imaxa{A} = i]$, where the probability is taken over the randomness of the utilities.
Let $\bp^{A} = (p^{A}_1, \dots, p^{A}_k)$.
\item Let $\CEF{A}$ be the partition of $N$ where each agent $a\in N$ is assigned to $\imaxa{A}$. 
Note that the outcome $(\CEF{A}, A)$ is always envy-free. 
\item Let $\bbP(N)$ denote the set of all partitions $(C_1, \dots, C_k)$ of~$N$ such that $|C_1| = \cdots = |C_k| = n/k$.
\item For $A \in \cA$, let $\cE_{A}$ denote the event that $\CEF{A} \in \bbP(N)$.
\end{itemize}

Since $\cD$ is non-atomic, a tie in utilities (i.e., $u_a(B) = u_a(B')$ for some agent $a$ and distinct bundles $B, B'$) occurs with probability zero. 
We thus have the following.

\begin{observation} \label{obs:simplified-prob}
The probability that an envy-free outcome exists is equal to $\Pr[\bigvee_{A \in \cA}  \cE_{A}]$.
\end{observation}

We continue with further preliminaries on probabilities.

\begin{itemize}
\item For a distribution $\cP$, we write $\supp(\cP)$ to denote its support. 
For $\Lambda \subseteq \supp(\cP)$, we write $\cP(\Lambda)$ as a shorthand for the measure (i.e., probability) of $\Lambda$ under $\cP$. 
When $\Lambda = \{\lambda\}$ has size one, we simply write $\cP(\lambda)$.
\item For $\Sigma \in \R^{d \times d}$, we write $\cN(\bzero, \Sigma)$ to denote the centered (multivariate) Gaussian distribution supported on $\R^d$ with covariance $\Sigma$.
\item Let $\Delta^{k-1}_n$ denote the discrete $(k-1)$-simplex $\{\bx \in \Z_{\geq 0}^k \mid \sum_{i \in [k]} x_i = n\}$.
\item For $n \in \N$ and $\bp \in \Delta^{k-1}$, the multinomial distribution $\multi{n}{\bp}$ is a distribution supported on $\Delta^{k-1}_n$ where $\multi{n}{\bp}\{\bx\} := \binom{n}{\bx} \prod_{i \in [k]} p_i^{x_i}$ for each $\bx \in \Delta^{k-1}_n$.
\end{itemize}

We do not include every proof here; all missing proofs can be found in the appendix.

\subsubsection{Non-Existence via First Moment Method.}

We begin by proving the non-existence result (\Cref{thm:lb-div-asympt}) via the first moment method. 
Specifically, in light of \Cref{obs:simplified-prob}, we compute $\Pr[\cE_{A}]$ for any allocation $A$ and apply the union bound. 
We start with a simple formula.
\begin{observation} \label{obs:event-prob-simplified}
For any allocation~$A$, we have $\Pr[\cE_{A}] = \multi{n}{\bp^{A}}\left\{\frac{n}{k} \cdot \bone\right\}$.
\end{observation}
\begin{proof}
Since the agents' utilities are independent, the distribution of the number of agents in $\CEF{A}$ is $\multi{n}{\bp^{A}}$.
The formula follows from the definition of~$\cE_{A}$.
\end{proof}

A standard bound on $\multi{n}{\bp^{A}}\left\{\frac{n}{k} \cdot \bone\right\}$ then yields the following lemma.

\begin{lemma} \label{lem:first-moment-main}
For any $A \in \cA$, we have $\Pr[\cE_{A}] \leq \frac{k^{k/2}}{(2\pi n)^{\frac{k-1}{2}}}$.
Moreover, for any $A \in \cAbal$, 
$\Pr[\cE_{A}] 
\geq \frac{1}{e^{k^2/n}} \cdot \frac{k^{k/2}}{(2\pi n)^{\frac{k-1}{2}}}$.
\end{lemma}

We can now prove \Cref{thm:lb-div-asympt} by taking the union bound.

\begin{proof}[Proof of \Cref{thm:lb-div-asympt}]
Using \Cref{obs:simplified-prob} and \Cref{lem:first-moment-main} and taking the union bound over all $A \in \cA$, we find that the probability that an envy-free outcome exists is at most
\begin{align*}
&k^m \cdot \frac{k^{k/2}}{(2\pi n)^{\frac{k-1}{2}}} \leq n^{\frac{\ln k}{4}} \cdot \frac{k^{k/2}}{n^{\frac{k-1}{2}}} = \frac{k^{k/2}}{n^{\frac{2k - \ln k - 2}{4}}},
\end{align*}
where we use the assumption that $m\le \frac{\ln n}{4}$.
The last quantity is $o(1)$ for any fixed $k \geq 2$.
\end{proof}

\subsubsection{Existence via Second Moment Method.}
Next, we turn to the existence result (\Cref{thm:existence-div}).
Recall that we have estimated the first moment $\Pr[\cE_{A}]$ in \Cref{lem:first-moment-main}. 
One might expect that when the sum of $\Pr[\cE_{A}]$ over all $A \in \cA$ far exceeds~$1$, our desired event $\bigvee_{A \in \cA}  \cE_{A}$ occurs with high probability.
However, this is not necessarily true for dependent events. 
To overcome this issue, we employ the \emph{second moment method}, which leverages the fact that if the events are nearly independent, in the sense that $\Pr[\cE_{A} \wedge \cE_{A'}] \leq (1 + o(1)) \Pr[\cE_{A}]\cdot\Pr[\cE_{A'}]$, then we can reach the desired conclusion. 
Unfortunately, such an inequality need not hold for all pairs of allocations $A, A'$. 
Indeed, when $A$ and~$A'$ are ``close'', $\cE_{A}$ and $\cE_{A'}$ can be highly correlated. 
Therefore, we will only apply the second-moment method on a set of allocations that are ``far away'' from one another. 
The exact definition of ``far away'' that we use is given below in \Cref{def:delta-ib}. 
Observe that if $A, A'$ are drawn independently at random from $\cAbal$, then $\E[|A_i \cap A'_j|] = m/k^2$ for any $i,j\in[k]$. 
Thus, the condition in \Cref{def:delta-ib} requires that $|A_i \cap A'_j|$ is close to its expectation.

\begin{definition} \label{def:delta-ib}
For $\delta\in[0,1]$, two balanced allocations $A = (A_1, \dots, A_k)$ and $A' = (A'_1, \dots, A'_k)$ are said to be \emph{$\delta$-intersection balanced ($\delta$-IB}) if $\frac{m(1 - \delta)}{k^2} \leq |A_i \cap A'_j| \leq \frac{m(1 + \delta)}{k^2}$ for all $i,j\in[k]$. 
We say that balanced allocations $A^{(1)}, \dots, A^{(D)}$ are \emph{$\delta$-IB} if every pair among them is $\delta$-IB.
\end{definition}

A standard concentration argument shows that if $m \geq \Theta_{k, \delta}(\log n)$, then $n^{2k}$ random balanced allocations satisfy $\delta$-IB with high probability, as stated below.

\begin{lemma} \label{lem:random-balanced}
For any $\delta \in (0, 1)$ and any $m \geq \frac{4k^4}{\delta^2} \cdot \ln n$ such that $m$ is divisible by $k$, let $A^{(1)}, \dots, A^{(n^{2k})}$ be balanced allocations drawn uniformly and independently at random. 
Then, with high probability, these allocations are $\delta$-IB.
\end{lemma}

The remainder of this section is largely devoted to bounding the second moment as stated in the following lemma, where $\sigma^2 > 0$ is the variance of $\cD$.

\begin{lemma} \label{lem:second-moment-main}
For any fixed $\beta \in (0, 1)$, if $\delta \leq \frac{\beta}{12k^4}$ and $m \geq \frac{48000 k^{13}}{\beta^2 \sigma^6}$, then for any pair of $\delta$-IB balanced allocations $A$ and $A'$, we have
\[
\Pr[\cE_{A'} \mid \cE_{A}] \leq \frac{1}{1 - \beta/2 - o(1)} \cdot \frac{k^{k/2}}{(2\pi n)^{\frac{k-1}{2}}}.\]

\end{lemma}

Before we prove \Cref{lem:second-moment-main}, let us first show how to use it to establish \Cref{thm:existence-div}.

\begin{proof}[Proof of \Cref{thm:existence-div}]
Let $C_{k, \beta, \cD} = \max\left\{\frac{4k^4}{\delta^2}, \frac{1920000k^{13}}{\beta^2\sigma^6} \right\}$, where $\delta = \frac{\beta}{12k^4}$.
The algorithm works as follows.
\begin{itemize}
\item Choose balanced allocations $A^{(1)}, \dots, A^{(n^{2k})}$ independently and uniformly at random.\footnote{For example, we can take a random permutation of the goods, and let each block of $m/k$ goods form a bundle.}
\item For $w \in [n^{2k}]$, check whether $\CEF{A^{(w)}} \in \bbP(N)$.
If so, output $(\CEF{A^{(w)}}, A^{(w)})$.
\end{itemize}
The algorithm runs in time $n^{O(k)} \cdot m$. 
To bound its success probability, note that by \Cref{lem:random-balanced}, $A^{(1)}, \dots, A^{(n^{2k})}$ are $\delta$-IB with high probability. 
We condition on this event henceforth.

Our algorithm succeeds when the event $\bigvee_{w \in [n^{2k}]} \cE_{A^{(w)}}$ occurs. 
Let $Z = |\{w \in [n^{2k}] \mid \cE_{A^{(w)}}\}|$ and $q = \Pr[\cE_{A^{(w)}}]$; note that $q$ does not depend on $w$ due to symmetry. 
Using the second moment method, the probability that the algorithm succeeds can be bounded as follows:
\begin{align*}
&\Pr\left[\bigvee_{w \in [n^{2k}]} \cE_{A^{(w)}}\right]
= \Pr[Z > 0] 
\geq \frac{\left(\E[Z]\right)^2}{\E[Z^2]} \\
&= \frac{1}{\frac{1}{n^{4k}} \sum_{w \in [n^{2k}]} \left(\frac{1}{q} + \sum_{v \in [n^{2k}] \setminus \{w\}} \frac{\Pr[\cE_{A^{(v)}} \mid \cE_{A^{(w)}}]}{q}\right)}.
\end{align*}
Applying \Cref{lem:first-moment-main,lem:second-moment-main}, the denominator is at most
\begin{align*}
&\frac{1}{n^{4k}} \sum_{w \in [n^{2k}]} \left(O\left(n^{\frac{k-1}{2}}\right) + \sum_{v \in [n^{2k}] \setminus \{w\}} \frac{1}{1 - \beta/2 -o(1)}\right) \\
&= \frac{1}{1 - \beta/2 - o(1)}.
\end{align*}
Thus, our algorithm succeeds with probability $1 - \beta/2 - o(1)$, which is at least $1 - \beta$ for any sufficiently large $n$.
\end{proof}
 
We now proceed to prove \Cref{lem:second-moment-main}.
Fix two balanced allocations $A, A'$ that are $\delta$-IB. 
For $i, i' \in [k]$, we let\footnote{$p_{(i, i')}$ is the same for every agent~$a$, so we omit~$a$ from the notation.} $p_{(i, i')} = \Pr[\imaxa{A} = i \wedge \imaxa{A'} = i']$.

\begin{lemma} \label{lem:cond-as-sum}
For each $a\in N$, let $X^a$ be a random variable on $\{\be_1, \dots, \be_k\}$ such that $\Pr[X^a = \be_{i'}] = k \cdot p_{(\lceil ka / n \rceil, i')}$ for each $i' \in [k]$.
Then, $\Pr[\cE_{A'} \mid \cE_{A}] = \Pr\left[X^1 + \cdots + X^n = \frac{n}{k}\cdot\bone\right]$.
\end{lemma}

\begin{proof}
By definition, $\Pr[\cE_{A'} \mid \cE_{A}] = \Pr[\CEF{A'} \in \bbP(N) \mid \CEF{A} \in \bbP(N)]$.
Due to the symmetry across agents and since $A$ is balanced, we may assume that $\CEF{A}$ assigns agent $a \in N$ to group $\lceil ka 
/ n \rceil$. 
Conditioned on this, we can take $X^a$ to be the random variable such that $X^a = \be_{\imaxa{A'}}$. 
Indeed, we have that 
\begin{align*}
\Pr[X^a = \be_{i'}] &= \Pr[\imaxa{A'} = i' \mid \imaxa{A} = \lceil ka / n\rceil] \\ &= \frac{p_{(\lceil ka / n\rceil, i')}}{\Pr[\imaxa{A} = \lceil ka / n\rceil]} = k \cdot p_{(\lceil ka / n \rceil, i')},
\end{align*}
where we use Bayes' law and the symmetry across groups, respectively. 
Finally, observe that $\CEF{A'} \in \bbP(N)$ is equivalent to $X^1 + \cdots + X^n = \frac{n}{k} \cdot \bone$.
\end{proof}

The rest of the proof of \Cref{lem:second-moment-main} can be divided into two parts: (i) showing that the values $p_{(i, i')}$ are all close to $1/k^2$, and (ii) bounding $\Pr\left[X^1 + \cdots + X^n = \frac{n}{k}\cdot\bone\right]$.

\paragraph{Part (i): Bounding $p_{(i, i')}$.}
We need the following multivariate generalization of the Berry–Esseen theorem, which is stated as Theorem~1.1 in the work of \citet{Raic19}.\footnote{
The version stated here follows from Theorem~1.1 of \citet{Raic19} by letting $X_i = \Sigma^{-\frac{1}{2}} W^i$, so that $\sum_{i=1}^T\var(X_i) = I_d$.
Also, we use the term $60d^{1/4}$, which is weaker than the one used by \citet{Raic19}.}

\begin{lemma}[\citealt{Raic19}]
\label{lem:mult-berry-esseen}
Let $W^1, \dots, W^T$ be independent random variables in $\R^d$ with mean zero, $H := W^1 + \cdots + W^T$, and $\Sigma := \cov{H} \in \R^{d \times d}$ be the covariance matrix of $H$. 
For any convex $\Lambda \subseteq \R^d$, it holds that $\left|\Pr[H \in \Lambda] - \cN(\bzero, \Sigma)\{\Lambda\}\right|
\leq 60 d^{1/4} \sum_{i \in [T]} \E[\|\Sigma^{-\frac{1}{2}} W^i\|_2^3].$
\end{lemma}

Let us now give the proof overview for this part. 
We let $d = k^2$ and implicitly associate\footnote{For example, $(j, j')$ can be associated with $k(j - 1) + j'$.} tuples $(j, j') \in [k]^2$ with elements in $[k^2]$ when writing the indices for readability.

Let $\mu$ denote the mean of $\cD$. 
For each $g \in M$, we define the random variable $W^g \in \R^{d}$ such that
\begin{align*}
W^g_{(j, j')} =
\begin{cases}
\left(u_a(g) - \mu\right) \cdot \frac{k}{\sigma\sqrt{m}} &\text{ if } g \in A_{j} \cap A'_{j'}; \\
0 &\text{ otherwise}
\end{cases}
\end{align*}
for all $j, j' \in [k]$; note that this distribution is the same regardless of~$a$.
Then, let $H = \sum_{g \in M} W^g$.
Moreover, for all $i, i' \in [k]$, let $\Lambda_{i,i'} \subseteq \R^d$ be the set of all vectors $\bv \in \R^d$ that satisfy the following constraints:
\begin{align*}
\sum_{\ell \in [k]} v_{(i,\ell)} &\geq \sum_{\ell \in [k]} v_{(j,\ell)} &\forall j \in [k]; \\ 
\sum_{\ell \in [k]} v_{(\ell,i')} &\geq \sum_{\ell \in [k]} v_{(\ell,j')} &\forall j' \in [k].
\end{align*}
One can check that when there are no ties in utilities, $\imaxa{A} =i$ and $\imaxa{A'} = i'$ if and only if $H \in \Lambda_{i, i'}$. 
In other words, $\Pr[H \in \Lambda_{i, i'}]$ is exactly $p_{(i, i')}$.
Since $\Lambda_{i,i'}$ is convex, \Cref{lem:mult-berry-esseen} implies that $p_{(i,i')}$ is close to $\cN(\bzero, \Sigma)\{\Lambda_{i, i'}\}$ for $\Sigma = \cov{H}$. 
Using the fact that $A, A'$ are $\delta$-IB and balanced, we can show that $\Sigma$ is also close to the identity matrix $I_{d}$. 
By applying a standard total variation distance bound between Gaussians, this also implies that $\cN(\bzero, \Sigma)\{\Lambda_{i, i'}\}$ and $\cN(\bzero, I_{d})\{\Lambda_{i, i'}\}$ are close. 
Due to symmetry, the latter is simply $\frac{1}{k^2}$. Thus, we can conclude that $p_{(i, i')}$ itself is close to $\frac{1}{k^2}$, as stated more formally below.
\begin{lemma} \label{lem:prob-close}
Let $\gamma \in (0, 1)$. Suppose that $\delta \le \frac{\gamma}{3k}$ and $m \geq \frac{30000 k^7}{\gamma^2 \sigma^6}$. Then, for all $i, i' \in [k]$, we have $\left|p_{(i, i')} - \frac{1}{k^2}\right| \le \gamma$.
\end{lemma}

\paragraph{Part (ii): Bounding $\Pr\left[X^1 + \cdots + X^n = \frac{n}{k}\cdot\bone\right]$.} We reinterpret $X^i$ so that it is a mixture distribution between the uniform distribution on $\{\be_1, \dots, \be_k\}$ and a ``leftover'' distribution. 
In other words, for each $X^i$, we can toss a (biased) coin and, based on the outcome, sample $X^i$ either from the uniform distribution or from the leftover distribution.
By a concentration bound, we show that $X^i$ is drawn from the uniform distribution for most indices $i$. 
For these $X^i$, their sum exactly follows the multinomial distribution, which we have a very good estimate on. 
From this, we can derive the following bound.

\begin{lemma} \label{lem:maxprob-main}
Let $\zeta \in (0, 0.5)$ and $\tX^1, \dots, \tX^n$ be any independent random variables on $\{\be_1, \dots, \be_k\}$ such that $|\Pr[\tX^a = \be_i] - \frac{1}{k}| \leq \frac{\zeta}{k}$ for all $a \in N$ and $i \in [k]$. 
If $n \geq \frac{k}{1 - 2\zeta}$, then $\Pr\left[\tX^1 + \cdots + \tX^n = \frac{n}{k}\cdot\bone\right]$ is at most
\begin{align*}
\exp\left(-2\zeta^2 n\right) + \frac{k^{k/2}}{(2\pi(1 - 2\zeta)n)^{\frac{k-1}{2}}}.
\end{align*}
\end{lemma}

Finally, combining Lemmas~\ref{lem:cond-as-sum}, \ref{lem:prob-close}, and \ref{lem:maxprob-main} yields \Cref{lem:second-moment-main}.

\subsection{Non-Divisible Case}
\label{sec:non-divisible}

Next, we consider the case where $m$ is not divisible by~$k$.
For this case, we prove that an envy-free outcome cannot exist with high probability unless $m = \Omega(\sqrt{n})$. 
This differs markedly from the divisible case, where $m = \omega(\log n)$ suffices for existence. 

\begin{theorem} \label{thm:lb-indiv-asympt}
For any fixed $k \geq 2$, if $m = o(\sqrt{n})$ and $m$ is not divisible by $k$, then with high probability, no envy-free outcome exists.
\end{theorem}

The proof of this non-existence result, like the divisible case (\Cref{thm:lb-div-asympt}), uses the first moment method. 
The key distinguishing property between the two cases is encapsulated in the lemma below, which states that for any allocation~$A$, there exists $i \in [k]$ such that $p_i^{A}$ is noticeably smaller than $\frac{1}{k}$; in particular, we may take an index~$i$ corresponding to a bundle with the smallest size.
This is in contrast to the divisible case, for which a balanced allocation gives $p_i^{A} = \frac{1}{k}$ for all $i \in [k]$.

\begin{lemma} \label{lem:prob-smallest-group}
If $m$ is not divisible by $k$, then for any $A \in \cA$, there exists $i \in [k]$ such that
$p_i^{A} \leq \frac{1}{k} - \frac{\alpha}{k^2\sqrt{m}},$
where $\alpha > 0$ is a constant depending only on $\cD$.
\end{lemma}

\Cref{lem:prob-smallest-group} implies the following upper bound on $\Pr[\cE_{A}]$.

\begin{lemma} \label{lem:prob-single-alloc-ef}
If $m$ is not divisible by $k$, then for any $A \in \cA$, we have $\Pr[\cE_{A}] \leq \sqrt{k} \cdot \exp\left(-\frac{\alpha^2}{k^4} \cdot \frac{n}{m}\right)$,
where $\alpha$ is the constant from \Cref{lem:prob-smallest-group}.
\end{lemma}

We now finish the proof in a similar manner as \Cref{thm:lb-div-asympt}.

\begin{proof}[Proof of \Cref{thm:lb-indiv-asympt}]
Using \Cref{obs:simplified-prob} and \Cref{lem:prob-single-alloc-ef} and taking the union bound over all $A \in \cA$, the probability that an envy-free outcome exists is at most $k^m \cdot \sqrt{k} \cdot \exp\left(-\frac{\alpha^2}{k^4} \cdot \frac{n}{m}\right)$,
which is $o(1)$ because $m = o(n/m)$.
\end{proof}

When $m = \Omega(n \log n)$, an envy-free outcome exists with high probability due to the result in the fixed-group setting by \citet{ManurangsiSu17}. 
Tightening the gap between $o(\sqrt{n})$ and $\Omega(n\log n)$ is an interesting question.

\section{Conclusion}

In this paper, we have studied fairness in the allocation of indivisible goods with variable groups, where a partition of the agents into groups can be chosen along with an allocation of the goods.
We demonstrated that the flexibility afforded by this model allows strong envy-freeness guarantees to be made in both worst-case and average-case scenarios.

Besides closing the asymptotic gap in the non-divisible case (\Cref{sec:non-divisible}) and extending the asymptotic analysis in \Cref{sec:EF-asymptotic} to the case where groups may have differing sizes, an intriguing direction for future work is to make connections between the variable-group model and the setting of \emph{hedonic games} \citep{AzizSa16}.
Specifically, in hedonic games, agents derive utilities from other agents assigned to the same group, and the objective is to find a desirable partition of the agents into groups; variants where the group sizes are fixed have also been considered \citep{BiloMoMo22,LiMiNi23}.
A generalization of both models would therefore be to permit preferences both over agents as well as over goods---this can reflect, e.g., group projects where resources are assigned to each group.\footnote{A similar model has been studied under the name \emph{generalized group activity selection problem} \citep{BiloFaFl19,FlamminiVa22}.}
It would be interesting to investigate fairness guarantees that can be made in this general setting.

\section*{Acknowledgments}

This work was partially supported by the Singapore Ministry of Education under grant number MOE-T2EP20221-0001, by JST FOREST Grant Number JPMJFR226O, and by an NUS Start-up Grant.
We would like to thank Erel Segal-Halevi for helpful discussions at the early stage, and the anonymous reviewers for their valuable comments.

\bibliography{aaai2026}

\begin{thebibliography}{51}
\providecommand{\natexlab}[1]{#1}

\bibitem[{Artin(1964)}]{Artin64}
Artin, E. 1964.
\newblock \emph{The Gamma Function}.
\newblock Athena Series. Holt, Rinehart and Winston.

\bibitem[{Aziz(2020)}]{Aziz20}
Aziz, H. 2020.
\newblock Developments in multi-agent fair allocation.
\newblock In \emph{Proceedings of the 34th AAAI Conference on Artificial Intelligence (AAAI)}, 13563--13568.

\bibitem[{Aziz and Savani(2016)}]{AzizSa16}
Aziz, H.; and Savani, R. 2016.
\newblock Hedonic games.
\newblock In Brandt, F.; Conitzer, V.; Endriss, U.; Lang, J.; and Procaccia, A.~D., eds., \emph{Handbook of Computational Social Choice}, chapter~15, 356--376. Cambridge University Press.

\bibitem[{Bai et~al.(2022)Bai, Feige, G\"{o}lz, and Procaccia}]{BaiFeGo22}
Bai, Y.; Feige, U.; G\"{o}lz, P.; and Procaccia, A.~D. 2022.
\newblock Fair allocations for smoothed utilities.
\newblock In \emph{Proceedings of the 23rd ACM Conference on Economics and Computation (EC)}, 436--465.

\bibitem[{Bai and G\"{o}lz(2022)}]{BaiGo22}
Bai, Y.; and G\"{o}lz, P. 2022.
\newblock Envy-free and {P}areto-optimal allocations for agents with asymmetric random valuations.
\newblock In \emph{Proceedings of the 31st International Joint Conference on Artificial Intelligence (IJCAI)}, 53--59.

\bibitem[{Barman et~al.(2025)Barman, Ebadian, Latifian, and Shah}]{BarmanEbLa25}
Barman, S.; Ebadian, S.; Latifian, M.; and Shah, N. 2025.
\newblock Fair division with market values.
\newblock In \emph{Proceedings of the 39th AAAI Conference on Artificial Intelligence (AAAI)}, 13589--13596.

\bibitem[{Bei et~al.(2022)Bei, Igarashi, Lu, and Suksompong}]{BeiIgLu22}
Bei, X.; Igarashi, A.; Lu, X.; and Suksompong, W. 2022.
\newblock The price of connectivity in fair division.
\newblock \emph{SIAM Journal on Discrete Mathematics}, 36(2): 1156--1186.

\bibitem[{Benad\`{e} et~al.(2024)Benad\`{e}, Halpern, Psomas, and Verma}]{BenadeHaPs24}
Benad\`{e}, G.; Halpern, D.; Psomas, A.; and Verma, P. 2024.
\newblock On the existence of envy-free allocations beyond additive valuations.
\newblock In \emph{Proceedings of the 25th ACM Conference on Economics and Computation (EC)}, 1287.

\bibitem[{Bil\`{o} et~al.(2022)Bil\`{o}, Caragiannis, Flammini, Igarashi, Monaco, Peters, Vinci, and Zwicker}]{BiloCaFl22}
Bil\`{o}, V.; Caragiannis, I.; Flammini, M.; Igarashi, A.; Monaco, G.; Peters, D.; Vinci, C.; and Zwicker, W.~S. 2022.
\newblock Almost envy-free allocations with connected bundles.
\newblock \emph{Games and Economic Behavior}, 131: 197--221.

\bibitem[{Bil\`{o} et~al.(2019)Bil\`{o}, Fanelli, Flammini, Monaco, and Moscardelli}]{BiloFaFl19}
Bil\`{o}, V.; Fanelli, A.; Flammini, M.; Monaco, G.; and Moscardelli, L. 2019.
\newblock Optimality and {N}ash stability in additive separable generalized group activity selection problems.
\newblock In \emph{Proceedings of the 28th International Joint Conference on Artificial Intelligence (IJCAI)}, 102--108.

\bibitem[{Bil{\`o}, Monaco, and Moscardelli(2022)}]{BiloMoMo22}
Bil{\`o}, V.; Monaco, G.; and Moscardelli, L. 2022.
\newblock Hedonic games with fixed-size coalitions.
\newblock In \emph{Proceedings of the 36th AAAI Conference on Artificial Intelligence (AAAI)}, 9287--9295.

\bibitem[{Bouveret et~al.(2017)Bouveret, Cechl\'{a}rov\'{a}, Elkind, Igarashi, and Peters}]{BouveretCeEl17}
Bouveret, S.; Cechl\'{a}rov\'{a}, K.; Elkind, E.; Igarashi, A.; and Peters, D. 2017.
\newblock Fair division of a graph.
\newblock In \emph{Proceedings of the 26th International Joint Conference on Artificial Intelligence (IJCAI)}, 135--141.

\bibitem[{Bouveret, Chevaleyre, and Maudet(2016)}]{BouveretChMa16}
Bouveret, S.; Chevaleyre, Y.; and Maudet, N. 2016.
\newblock Fair allocation of indivisible goods.
\newblock In Brandt, F.; Conitzer, V.; Endriss, U.; Lang, J.; and Procaccia, A.~D., eds., \emph{Handbook of Computational Social Choice}, chapter~12, 284--310. Cambridge University Press.

\bibitem[{Bu et~al.(2023)Bu, Li, Liu, Song, and Tao}]{BuLiLi23}
Bu, X.; Li, Z.; Liu, S.; Song, J.; and Tao, B. 2023.
\newblock Fair division with allocator's preference.
\newblock In \emph{Proceedings of the 19th International Conference on Web and Internet Economics (WINE)}, 77--94.

\bibitem[{Caragiannis, Larsen, and Shyam(2025)}]{CaragiannisLaSh25}
Caragiannis, I.; Larsen, K.~G.; and Shyam, S. 2025.
\newblock A new lower bound for multi-color discrepancy with applications to fair division.
\newblock In \emph{Proceedings of the 18th International Symposium on Algorithmic Game Theory (SAGT)}, 228--246.

\bibitem[{Deng, Qi, and Saberi(2012)}]{DengQiSa12}
Deng, X.; Qi, Q.; and Saberi, A. 2012.
\newblock Algorithmic solutions for envy-free cake cutting.
\newblock \emph{Operations Research}, 60(6): 1461--1476.

\bibitem[{Devroye, Mehrabian, and Reddad(2023)}]{DevroyeMeRe23}
Devroye, L.; Mehrabian, A.; and Reddad, T. 2023.
\newblock The total variation distance between high-dimensional Gaussians with the same mean.
\newblock \emph{CoRR}, abs/1810.08693v7.

\bibitem[{Dickerson et~al.(2014)Dickerson, Goldman, Karp, Procaccia, and Sandholm}]{DickersonGoKa14}
Dickerson, J.~P.; Goldman, J.; Karp, J.; Procaccia, A.~D.; and Sandholm, T. 2014.
\newblock The computational rise and fall of fairness.
\newblock In \emph{Proceedings of the 28th AAAI Conference on Artificial Intelligence (AAAI)}, 1405--1411.

\bibitem[{Flammini and Varricchio(2022)}]{FlamminiVa22}
Flammini, M.; and Varricchio, G. 2022.
\newblock Approximate strategyproof mechanisms for the additively separable group activity selection problem.
\newblock In \emph{Proceedings of the 31st International Joint Conference on Artificial Intelligence (IJCAI)}, 300--306.

\bibitem[{Ghodsi et~al.(2018)Ghodsi, Latifian, Mohammadi, Moradian, and Seddighin}]{GhodsiLaMo18}
Ghodsi, M.; Latifian, M.; Mohammadi, A.; Moradian, S.; and Seddighin, M. 2018.
\newblock Rent division among groups.
\newblock In \emph{Proceedings of the 12th International Conference on Combinatorial Optimization and Applications (COCOA)}, 577--591.

\bibitem[{G\"{o}lz and Yaghoubizade(2026)}]{GolzYa26}
G\"{o}lz, P.; and Yaghoubizade, H. 2026.
\newblock Fair division among couples and small groups.
\newblock In \emph{Proceedings of the 40th AAAI Conference on Artificial Intelligence (AAAI)}.
\newblock Forthcoming.

\bibitem[{Igarashi(2023)}]{Igarashi23}
Igarashi, A. 2023.
\newblock How to cut a discrete cake fairly.
\newblock In \emph{Proceedings of the 37th AAAI Conference on Artificial Intelligence (AAAI)}, 5681--5688.

\bibitem[{Igarashi and Meunier(2025)}]{IgarashiMe24}
Igarashi, A.; and Meunier, F. 2025.
\newblock Envy-free division of multilayered cakes.
\newblock \emph{Mathematics of Operations Research}, 50(3): 2261--2286.

\bibitem[{Kawase, Roy, and Sanpui(2025)}]{KawaseRoSa25}
Kawase, Y.; Roy, B.; and Sanpui, M.~A. 2025.
\newblock Simultaneously fair allocation of indivisible items across multiple dimensions.
\newblock In \emph{Proceedings of the 45th IARCS Annual Conference on Foundations of Software Technology and Theoretical Computer Science (FSTTCS)}.
\newblock Forthcoming.

\bibitem[{Kurokawa, Procaccia, and Wang(2016)}]{KurokawaPrWa16}
Kurokawa, D.; Procaccia, A.~D.; and Wang, J. 2016.
\newblock When can the maximin share guarantee be guaranteed?
\newblock In \emph{Proceedings of the 30th AAAI Conference on Artificial Intelligence (AAAI)}, 523--529.

\bibitem[{Kyropoulou, Suksompong, and Voudouris(2020)}]{KyropoulouSuVo20}
Kyropoulou, M.; Suksompong, W.; and Voudouris, A.~A. 2020.
\newblock Almost envy-freeness in group resource allocation.
\newblock \emph{Theoretical Computer Science}, 841: 110--123.

\bibitem[{Li et~al.(2023)Li, Micha, Nikolov, and Shah}]{LiMiNi23}
Li, L.; Micha, E.; Nikolov, A.; and Shah, N. 2023.
\newblock Partitioning friends fairly.
\newblock In \emph{Proceedings of the 37th AAAI Conference on Artificial Intelligence (AAAI)}, 5747--5754.

\bibitem[{Lipton et~al.(2004)Lipton, Markakis, Mossel, and Saberi}]{LiptonMaMo04}
Lipton, R.~J.; Markakis, E.; Mossel, E.; and Saberi, A. 2004.
\newblock On approximately fair allocations of indivisible goods.
\newblock In \emph{Proceedings of the 5th ACM Conference on Electronic Commerce (EC)}, 125--131.

\bibitem[{Manurangsi and Meka(2026)}]{ManurangsiMe25}
Manurangsi, P.; and Meka, R. 2026.
\newblock Tight lower bound for multicolor discrepancy.
\newblock In \emph{Proceedings of the 9th SIAM Symposium on Simplicity in Algorithms (SOSA)}.
\newblock Forthcoming.

\bibitem[{Manurangsi and Suksompong(2017)}]{ManurangsiSu17}
Manurangsi, P.; and Suksompong, W. 2017.
\newblock Asymptotic existence of fair divisions for groups.
\newblock \emph{Mathematical Social Sciences}, 89: 100--108.

\bibitem[{Manurangsi and Suksompong(2020)}]{ManurangsiSu20}
Manurangsi, P.; and Suksompong, W. 2020.
\newblock When do envy-free allocations exist?
\newblock \emph{{SIAM} Journal on Discrete Mathematics}, 34(3): 1505--1521.

\bibitem[{Manurangsi and Suksompong(2021)}]{ManurangsiSu21}
Manurangsi, P.; and Suksompong, W. 2021.
\newblock Closing gaps in asymptotic fair division.
\newblock \emph{{SIAM} Journal on Discrete Mathematics}, 35(2): 668--706.

\bibitem[{Manurangsi and Suksompong(2022)}]{ManurangsiSu22}
Manurangsi, P.; and Suksompong, W. 2022.
\newblock Almost envy-freeness for groups: Improved bounds via discrepancy theory.
\newblock \emph{Theoretical Computer Science}, 930: 179--195.

\bibitem[{Manurangsi and Suksompong(2025{\natexlab{a}})}]{ManurangsiSu25}
Manurangsi, P.; and Suksompong, W. 2025{\natexlab{a}}.
\newblock Asymptotic fair division: Chores are easier than goods.
\newblock In \emph{Proceedings of the 34th International Joint Conference on Artificial Intelligence (IJCAI)}, 3988--3995.

\bibitem[{Manurangsi and Suksompong(2025{\natexlab{b}})}]{ManurangsiSu24}
Manurangsi, P.; and Suksompong, W. 2025{\natexlab{b}}.
\newblock Ordinal maximin guarantees for group fair division.
\newblock \emph{Theoretical Computer Science}, 1036: 115151.

\bibitem[{Manurangsi, Suksompong, and Yokoyama(2025)}]{ManurangsiSuYo25}
Manurangsi, P.; Suksompong, W.; and Yokoyama, T. 2025.
\newblock Asymptotic analysis of weighted fair division.
\newblock \emph{Theoretical Computer Science}, 1054: 115533.

\bibitem[{Markakis(2017)}]{Markakis17}
Markakis, E. 2017.
\newblock Approximation algorithms and hardness results for fair division.
\newblock In Endriss, U., ed., \emph{Trends in Computational Social Choice}, chapter~12, 231--247. AI Access.

\bibitem[{Meunier and Su(2019)}]{MeunierSu19}
Meunier, F.; and Su, F.~E. 2019.
\newblock Multilabeled versions of {S}perner's and {F}an's lemmas and applications.
\newblock \emph{SIAM Journal on Applied Algebra and Geometry}, 3(3): 391--411.

\bibitem[{Postnikov and Yudin(1988)}]{PostnikovYu88}
Postnikov, A.~G.; and Yudin, A.~A. 1988.
\newblock Estimating the maximum probability of a sum of independent vectors.
\newblock \emph{Theory of Probability \& Its Applications}, 32(2): 331--334.

\bibitem[{Raič(2019)}]{Raic19}
Raič, M. 2019.
\newblock A multivariate {B}erry--{E}sseen theorem with explicit constants.
\newblock \emph{Bernoulli}, 25(4A): 2824--2853.

\bibitem[{Robbins(1955)}]{Robbins55}
Robbins, H. 1955.
\newblock A remark on {S}tirling's formula.
\newblock \emph{American Mathematical Monthly}, 62(1): 26--29.

\bibitem[{Segal-Halevi and Nitzan(2019)}]{SegalhaleviNi19}
Segal-Halevi, E.; and Nitzan, S. 2019.
\newblock Envy-free cake-cutting among families.
\newblock \emph{Social Choice and Welfare}, 53(4): 709--740.

\bibitem[{Segal-Halevi and Suksompong(2021)}]{SegalhaleviSu21}
Segal-Halevi, E.; and Suksompong, W. 2021.
\newblock How to cut a cake fairly: a generalization to groups.
\newblock \emph{American Mathematical Monthly}, 128(1): 79--83.

\bibitem[{Segal-Halevi and Suksompong(2023)}]{SegalhaleviSu23}
Segal-Halevi, E.; and Suksompong, W. 2023.
\newblock Cutting a cake fairly for groups revisited.
\newblock \emph{American Mathematical Monthly}, 130(3): 203--213.

\bibitem[{Serfling(1974)}]{Serfling74}
Serfling, R.~J. 1974.
\newblock Probability inequalities for the sum in sampling without replacement.
\newblock \emph{The Annals of Statistics}, 2(1): 39--48.

\bibitem[{Sperner(1928)}]{Sperner28}
Sperner, E. 1928.
\newblock Neuer Beweis für die Invarianz der Dimensionszahl und des Gebietes.
\newblock \emph{Abhandlungen aus dem Mathematischen Seminar der Universität Hamburg}, 6: 265--272.

\bibitem[{Su(1999)}]{Su99}
Su, F.~E. 1999.
\newblock Rental harmony: {S}perner's lemma in fair division.
\newblock \emph{American Mathematical Monthly}, 106(10): 930--942.

\bibitem[{Suksompong(2016)}]{Suksompong16}
Suksompong, W. 2016.
\newblock Asymptotic existence of proportionally fair allocations.
\newblock \emph{Mathematical Social Sciences}, 81: 62--65.

\bibitem[{Ushakov(1986)}]{Ushakov86}
Ushakov, N.~G. 1986.
\newblock Upper estimates of maximum probability for sums of independent random vectors.
\newblock \emph{Theory of Probability \& Its Applications}, 30(1): 38--49.

\bibitem[{Walsh(2020)}]{Walsh20}
Walsh, T. 2020.
\newblock Fair division: the computer scientist’s perspective.
\newblock In \emph{Proceedings of the 29th International Joint Conference on Artificial Intelligence (IJCAI)}, 4966--4972.

\bibitem[{Yokoyama and Igarashi(2025)}]{YokoyamaIg25}
Yokoyama, T.; and Igarashi, A. 2025.
\newblock Asymptotic existence of class envy-free matchings.
\newblock In \emph{Proceedings of the 24th International Conference on Autonomous Agents and Multiagent Systems (AAMAS)}, 2244--2252.

\end{thebibliography}

\appendix

\section{Surjectivity of $\of$ in the proof of \Cref{lem:main-coloring-to-assignment}}
\label{app:surjective}

We now show the surjectivity of $\of$ as defined in the proof of \Cref{lem:main-coloring-to-assignment}.
To this end, we will use the following lemma, which is stated as Lemma~2.8 in the work of \citet{MeunierSu19}.

\begin{lemma}[\citealt{MeunierSu19}] \label{lem:surjective}
Any continuous map $q$ from a polytope to itself, which satisfies $q(F) \subseteq F$ for
every face $F$ of the polytope, is surjective.
\end{lemma}

\begin{lemma}
Let $\of$ be as defined in the proof of \Cref{lem:main-coloring-to-assignment}. Then, $\of$ is surjective.
\end{lemma}

\begin{proof}
Let $\psi: \Delta^{k-1} \to S_m$ denote the canonical mapping from $\Delta^{k-1}$ to $S_m$, i.e.,
\begin{align*}
\psi(\bfy)_i = \frac{1}{2} + m \cdot \sum_{j=1}^{i} y_j & &\forall i \in [k - 1].
\end{align*}

Thus, to show that $\of$ is surjective, it suffices to show that $q :=  \of \circ \psi$ is surjective. 
By \Cref{lem:surjective}, we only need to show that $q$ satisfies the condition of the lemma. 
For each $a \in N$, define $q^{(a)} := \of^{(a)} \circ \psi$. 
Since $q = \frac{1}{n} \sum_{a \in N} q^{(a)}$, it is in turn sufficient to show that each $q^{(a)}$ satisfies the condition of \Cref{lem:surjective}.

Let us now fix $a \in N$, a face $F$ of $\Delta^{k-1}$, and $\bfy \in F$. 
Consider $\bfx = \psi(\bfy)$. 
By definition of $\of^{(a)}$, there exist\footnote{Note that $k'$ can be strictly less than $k$ if $\bfx$ belong to a face of an elementary simplex (since we only consider those $\bfx^h$ with coefficient strictly greater than zero).} $\bfx^1, \dots, \bfx^{k'} \in V(\ttT_{\half})$ and $\alpha^1, \dots, \alpha^{k'} > 0$ such that
\begin{align*}
\sum_{h \in [k']} \alpha^{h} &= 1; \\
\sum_{h \in [k']} \alpha^{h} \bfx^h &= \bfx; \\
\sum_{h \in [k']} \alpha^h f^{(a)}(\bfx^h) &= \of^{(a)}(\bfx).
\end{align*}
Thus, to show that $q^{(a)}(\bfy) = \of^{(a)}(\bfx)$ belongs to $F$, it suffices to argue that $f^{(a)}(\bfx^h)$ belongs to $F$ for all $h \in [k']$. 

Since $F$ is a face of $\Delta^{k-1}$, there exists $S \subseteq [k]$ such that $F$ is the span of $\{\bfe^{j} \mid j \in S\}$. 
From the definition of $\psi$, this implies that $x_i = x_{i-1}$ for all $i \notin S$. Since $\bfx$ is a strictly positive convex combination of $\bfx^1, \dots, \bfx^h$, we also have that $x^h_i = x^h_{i - 1}$ for all $h \in [k']$ and $i \notin S$. 
This means that $\bfx^h$ belongs to the span of $\{\psi(\bfe^{j}) \mid j \in S\}$; note also that each $\psi(\bfe^{j})$ is a vertex of $V(\ttT_{\half})$. 
Since $\lambda_a$ is proper, we have that $\lambda_a(\bfx^h) \in S$. It follows that $f^{(a)}(\bfx^h) = \bfe^{\lambda_a(\bfx^h)}$ belongs to $F$, as desired.
\end{proof}

\section{Deferred Proofs from \Cref{sec:EF-asymptotic}}

\subsection{Additional Preliminaries}

We introduce some additional notation.
\begin{itemize}
\item For $p \in [0, 1]$ and $r \in \N$, let $\ber{p}$ and $\bino{r}{p}$ denote the Bernoulli and binomial distributions with success probability $p$, respectively. 
\item We use $\mathrm{D_{KL}}$ and $\mathrm{D_{TV}}$ to denote the KL divergence and the total variation (TV) distance of two distributions, respectively.
Recall that the TV distance of two distributions $\cP_1, \cP_2$ is defined by
\begin{align*}
\tv{\cP_1}{\cP_2} = \sup_{\Lambda \subseteq \supp(\cP_1) \cup \supp(\cP_2)} |\cP_1(\Lambda) - \cP_2(\Lambda)|.
\end{align*}
Since we will only use the KL divergence on discrete distributions, we only state the definition for such distributions:
\begin{align*}
\kl{\cP_1}{\cP_2} = \sum_{\lambda \in \supp(\cP_1)} \cP_1(\lambda) \cdot \ln\left(\frac{\cP_1(\lambda)}{\cP_2(\lambda)}\right).
\end{align*}
\end{itemize}

We will also use Pinsker’s inequality, which relates the TV distance and the KL divergence.
\begin{lemma}[Pinsker’s inequality]
\label{lem:pinsker}
For any distributions $\cP_1$ and $\cP_2$, it holds that $\tv{\cP_1}{\cP_2} \leq \sqrt{\frac{1}{2} \cdot \kl{\cP_1}{\cP_2}}$.
\end{lemma}

\subsection{Multinomial Probability Mass Estimate} \label{app:stirling-mult-equip}

The following estimate of the probability mass of the multinomial distribution
follows from Stirling’s approximation of factorials. 
It will be useful in the proof of both the existential and non-existential results.

\begin{lemma} \label{cor:lem:stirling-mult-equip}
For any $n, k \in \N$ such that $n$ is divisible by $k$ and any $\bp \in \Delta^{k-1}$, there exists $\tr_{n, k} \in [0, \frac{k^2}{n}]$ such that $$\multi{n}{\bp}\left\{\frac{n}{k} \cdot \bone\right\} = \frac{1}{e^{\tr_{n, k}}} \cdot \frac{k^{k/2}}{(2\pi n)^{\frac{k-1}{2}}} \cdot \frac{1}{e^{n \cdot \kl{\frac{\bone}{k}}{\bp}}}.$$
\end{lemma}

To prove \Cref{cor:lem:stirling-mult-equip}, we will need the following precise version of Stirling's approximation.

\begin{lemma}[\citealt{Robbins55}] \label{lem:stirling}
For each positive integer $n$, it holds that 
\begin{align*}
n! = \sqrt{2\pi n} \left(\frac{n}{e}\right)^n e^{r_n},
\end{align*}
where $\frac{1}{12n + 1} < r_n < \frac{1}{12n}.$
\end{lemma}

\Cref{lem:stirling} yields a relatively simple formula for the probability mass function of the multinomial distribution.
\begin{lemma} \label{lem:stirling-mult}
For any $n \in \N$, $\bp \in \Delta^{k-1}$, and $\bx \in \Delta^{k-1}_n$ such that $x_i > 0$ for all $i \in [k]$, the probability $\multi{n}{\bp}\{\bx\}$ is equal to
\begin{align*}
\frac{1}{e^{r_{\bx}}} \cdot \frac{\sqrt{n}}{(2\pi)^{\frac{k-1}{2}} \sqrt{\prod_{i \in [k]} x_i}} \cdot \frac{1}{e^{n \cdot \kl{\frac{\bx}{n}}{\bp}}},
\end{align*}
where
$$0 \leq r_{\bx} \leq \frac{k}{\min_{i \in [k]} x_i}.$$
\end{lemma}
\begin{proof}
Letting $r_n$ be as in \Cref{lem:stirling}, we have
\begin{align*}
&\multi{n}{\bp}\{\bx\} = \binom{n}{\bx} \prod_{i \in [k]} p_i^{x_i} \\
&= \frac{\sqrt{2\pi n} \left(\frac{n}{e}\right)^n e^{r_n}}{\prod_{i \in [k]} (\sqrt{2\pi x_i} \left(\frac{x_i}{e}\right)^{x_i} e^{r_{x_i}})} \prod_{i \in [k]} p_i^{x_i} \\
&= \frac{1}{e^{r_{\bx}}} \cdot \frac{\sqrt{n}}{(2\pi)^{\frac{k-1}{2}} \sqrt{\prod_{i \in [k]} x_i}} \cdot \prod_{i \in [k]} \left(\frac{p_i}{x_i/n}\right)^{x_i} \\
&= \frac{1}{e^{r_{\bx}}} \cdot \frac{\sqrt{n}}{(2\pi)^{\frac{k-1}{2}} \sqrt{\prod_{i \in [k]} x_i}} \cdot e^{-n \cdot \kl{\frac{\bx}{n}}{\bp}},
\end{align*}
where $r_{\bx} := \sum_{i\in[k]} r_{x_i} - r_n $. From the bounds of $r_n$ in \Cref{lem:stirling}, it follows easily that $0 \leq r_{\bx} \leq \frac{k}{\min_{i \in [k]} x_i}$.
\end{proof}

\Cref{cor:lem:stirling-mult-equip} is then an immediate corollary of \Cref{lem:stirling-mult} (with $\bx = \frac{n}{k} \cdot \bone$ and $\tr_{n, k} = r_{\frac{n}{k} \cdot \bone}$).

\subsection{Proof of \Cref{lem:first-moment-main}}

The upper bound follows directly from \Cref{obs:event-prob-simplified}, \Cref{cor:lem:stirling-mult-equip}, and non-negativity of the KL-divergence.

For the lower bound, observe that due to the symmetry over bundles when $A \in \cAbal$, we have $p^{A}_1 = \cdots = p^{A}_k = \frac{1}{k}$, and so $\kl{\frac{\bone}{k}}{\bp} = 0$. 
Again, \Cref{obs:event-prob-simplified} and \Cref{cor:lem:stirling-mult-equip} then imply the bound.

\subsection{Proof of \Cref{lem:random-balanced}}
\label{app:random-balanced}

To prove \Cref{lem:random-balanced}, we will use the standard Hoeffding bound for sampling without replacement.
The following lemma is stated as Corollary 1.1 in the work of \citet{Serfling74}.\footnote{In fact, Serfling's version is slightly stronger because it contains the term $1-f_n^*$.}

\begin{lemma}[\citealt{Serfling74}] \label{lem:hoeffding}
Let $X_1, \dots, X_t$ be random variables drawn from some list $x_1, \dots, x_T \in \{0, 1\}$ without replacement.
Let $S = X_1 + \cdots + X_t$ and $\mu = \E[S]$. 
For any $\theta > 0$, we have
\begin{align*}
\Pr[S - \mu \geq \theta] &\leq \exp\left(-2\theta^2 / t\right); \\
\Pr[S - \mu \leq -\theta] &\leq \exp\left(-2\theta^2 / t\right).
\end{align*}
\end{lemma}

\begin{proof}[Proof of \Cref{lem:random-balanced}]
To compute the probability that $A^{(1)}, \dots, A^{(n^{2k})}$ are $\delta$-IB, let us fix $A^{(v)}$ for some $v \in [n^{2k}]$. 
Then, consider $A^{(w)}$ for $w \in [n^{2k}] \setminus \{v\}$ that is picked at random, and any $i, j \in [k]$. 
Once $A^{(v)}_i$ is fixed, we let $x_g := \bone[g \in A^{(v)}_i]$ for each $g \in M$, and view these values as in \Cref{lem:hoeffding}. 
Then, $A^{(w)}_j = \{g_1, \dots, g_{m/k}\}$ is chosen uniformly at random among all subsets of $M$ of size $m/k$; let $X_1 = x_{g_1}, \dots, X_{m/k} = x_{g_{m/k}}$. 
Observe that $S := X_1 + \cdots + X_{m/k}$ is equal to $|A^{(v)}_i \cap A^{(w)}_j|$ and $\E[S] = m/k^2$.
As such, we may apply \Cref{lem:hoeffding} with $\theta = \frac{\delta m}{k^2}$ to conclude that
\begin{align*}
&\Pr\left[|A^{(v)}_i \cap A^{(w)}_j| \notin \left[\frac{m(1 - \delta)}{k^2}, \frac{m(1 + \delta)}{k^2}\right] \right] \\
&\leq 2\exp\left(-2\delta^2 m / k^3\right) \\
&\leq \frac{2}{n^{8k}},
\end{align*}
where the last inequality is due to $m \geq \frac{4k^4}{\delta^2} \cdot \ln n$.

Taking the union bound over all distinct $v, w \in [n^{2k}]$ and all $i, j \in [k]$, we have $\Pr[A^{(1)}, \dots, A^{(n^{2k})} \text{ are } \delta\text{-IB}] \geq 1 - \frac{2k^2n^{4k}}{n^{8k}} = 1 - o(1)$ as desired.
\end{proof}

\subsection{Proof of \Cref{lem:second-moment-main}}

Let $\zeta = \frac{\beta}{4k}$ and $\gamma = \frac{\zeta}{k^2} = \frac{\beta}{4k^3}$. 
Applying \Cref{lem:cond-as-sum} yields
$\Pr[\cE_{A'} \mid \cE_{A}]$ 
$\leq \Pr\left[X^1 + \cdots + X^n = \frac{n}{k}\cdot\bone\right]$, where $X^1, \dots, X^n$ are as defined there.
By \Cref{lem:prob-close}, $X^1, \dots, X^n$ satisfy the condition of \Cref{lem:maxprob-main}; in particular, since we assume that $k$ is fixed and $n$ grows, the condition $n \geq \frac{k}{1 - 2\zeta}$ is met.
Hence,
\begin{align*}
\Pr[\cE_{A'} \mid \cE_{A}] 
&\leq \Pr\left[X^1 + \cdots + X^n = \frac{n}{k}\cdot\bone\right]\\
&\leq\exp\left(-2\zeta^2 n\right) + \frac{k^{k/2}}{(2\pi(1 - 2\zeta)n)^{\frac{k-1}{2}}} \\
&\leq o\left(\frac{1}{n^{\frac{k-1}{2}}}\right) + \frac{1}{(1 - 2\zeta)^k} \cdot \frac{k^{k/2}}{(2\pi n)^{\frac{k-1}{2}}} \\
&\leq o\left(\frac{1}{n^{\frac{k-1}{2}}}\right) + \frac{1}{1 - \beta/2} \cdot \frac{k^{k/2}}{(2\pi n)^{\frac{k-1}{2}}} \\
&\leq \frac{1}{1 - \beta/2 - o(1)} \cdot \frac{k^{k/2}}{(2\pi n)^{\frac{k-1}{2}}},
\end{align*}
where we use Bernoulli's inequality for the penultimate inequality. 

\subsection{Proof of \Cref{lem:prob-close}} 

To prove \Cref{lem:prob-close}, we need the following bound on the total variation distance between two centered Gaussians, stated as Theorem~1.1 in the work of \citet{DevroyeMeRe23}. 

\begin{lemma}[\citealt{DevroyeMeRe23}] \label{lem:tv-gaussian}
For positive definite matrices $\Sigma, \tSigma \in \R^{d \times d}$, let $\lambda_1, \dots, \lambda_d$ be the eigenvalues of $\tSigma^{-1}\Sigma - I_d$. Then, 
\begin{align*}
\tv{\cN(\bzero, \Sigma)}{\cN(\bzero, \tSigma)} \leq \frac{3}{2} \sqrt{\sum_{i=1}^d \lambda_i^2}.
\end{align*}
\end{lemma}

We are now ready to show that the probabilities $p_{(i,i)}$ are all close to $1/k^2$.

\begin{proof}[Proof of \Cref{lem:prob-close}]
Let $\mu$ and $\sigma^2$ denote the mean and variance of $\cD$, respectively. 
For each $g \in M$, let $W^g \in \R^{k^2}$ be a random variable such that
\begin{align*}
W^g_{(j, j')} =
\begin{cases}
\left(u_a(g) - \mu\right) \cdot \frac{k}{\sigma\sqrt{m}} &\text{ if } g \in A_{j} \cap A'_{j'}; \\
0 &\text{ otherwise}
\end{cases}
\end{align*}
for all $j, j' \in [k]$; note that this definition is the same regardless of~$a$.
Observe also that the random variables $W^g$ are independent for different $g$ (since the utilities are independent) and have mean zero. 
Next, let $H = \sum_{g \in M} W^g$. 
It can be seen that\footnote{Recall that, for $\bu \in \R^d$, $\diag(\bu) \in \R^{d \times d}$ denotes the diagonal matrix whose diagonal entries are $\bu$ and remaining entries are zero.} $\cov{H} = \diag((\sigma^2_{(j,j')})_{j,j'\in[k]})$, where $\sigma_{(j, j')} = \sqrt{\frac{k^2}{m} \cdot |A_{j} \cap A'_{j'}|}$; we write $\Sigma$ as a shorthand for $\diag((\sigma^2_{(j,j')})_{j,j'\in[k]})$. 
Finally, let $\Lambda_{i,i'}$ denote all vectors $\bv \in \R^{k^2}$ that satisfy 
\begin{align*}
\sum_{\ell \in [k]} v_{(i,\ell)} \geq \sum_{\ell \in [k]} v_{(j,\ell)} \text{ and } \sum_{\ell \in [k]} v_{(\ell,i')} \geq \sum_{\ell \in [k]} v_{(\ell,j')} 
\end{align*}
for all $j,j'\in [k]$.
Since $\Lambda_{i, i'}$ is convex, we can apply \Cref{lem:mult-berry-esseen} to get
\begin{align*}
&|\Pr[H \in \Lambda_{i, i'}] - \cN(0, \Sigma)\{\Lambda_{i, i'}\}| \\
&\leq 60 k^{1/2} \sum_{g \in M} \E[\|\Sigma^{-1/2} W^g\|_2^3].
\end{align*}

Now, notice that for any $j,j'\in[k]$,
\begin{align*}
&(\Sigma^{-1/2} W^g)_{(j, j')}\\
&=\begin{cases}
\left(u_a(g) - \mu\right) \cdot \frac{1}{\sigma\sqrt{|A_{j} \cap A'_{j'}|}} &\text{ if } g \in A_{j} \cap A'_{j'}; \\
0 &\text{ otherwise}.
\end{cases}
\end{align*}
Thus, since $u_a(g) \in [0, 1]$, we have $\|\Sigma^{-1/2} W^g\|_2 \leq \frac{1}{\sigma\sqrt{|A_{j} \cap A'_{j'}|}}$ for $g \in A_j \cap A'_{j'}$. 
Plugging this into the inequality above yields
\begin{align}
&|\Pr[H \in \Lambda_{i, i'}] - \cN(0, \Sigma)\{\Lambda_{i, i'}\}| \nonumber \\
&\leq 60 k^{1/2} \sum_{j,j' \in [k]} \sum_{g \in A_j \cap A'_{j'}} \frac{1}{\sigma^3 |A_j \cap A'_{j'}|^{1.5}} \nonumber \\
&= 60 k^{1/2} \sum_{j,j' \in [k]} \frac{1}{\sigma^3 \sqrt{|A_j \cap A'_{j'}|}} \nonumber \\
&\leq 60 k^{1/2} \sum_{j,j' \in [k]} \frac{1}{\sigma^3 \sqrt{\frac{m(1 - \delta)}{k^2}}} \nonumber \\
&= \sqrt{\frac{3600 k^7}{\sigma^6} \cdot \frac{1}{m(1 - \delta)}} \nonumber \\
&\leq 0.5\gamma, \label{eq:mult-berry-esseen-app}
\end{align}
where the second inequality is due to our assumption that $A, A'$ are $\delta$-IB and the last inequality follows from our assumption on $m$ and $k$.

Next, we bound the TV distance between $\cN(0, \Sigma)$ and $\cN(0, I_{k^2})$ via \Cref{lem:tv-gaussian}. 
To this end, observe that $\Sigma - I_{k^2} = \diag\left(\left(\frac{k^2}{m} \cdot |A_j \cap A'_{j'}| - 1\right)_{j,j'\in[k]}\right)$. 
Thus, the eigenvalues of $\Sigma - I_d$ are $\frac{k^2}{m} \cdot |A_j \cap A'_{j'}| - 1$ for all $j, j' \in [k]$. Since $A, A'$ are $\delta$-IB, we have
\begin{align*}
\left(\frac{k^2}{m} \cdot |A_j \cap A'_{j'}| - 1\right)^2 
&\leq \delta^2.
\end{align*}
Plugging this into \Cref{lem:tv-gaussian}, we conclude that
\begin{align} \label{eq:tv-gaussian-spherical}
\tv{\cN(\bzero, \Sigma)}{\cN(\bzero, I_{k^2})} \leq \frac{3}{2} \cdot k\delta \leq 0.5\gamma.
\end{align}

Therefore, by the triangle inequality, we have
\begin{align*}
&\left|\Pr[H \in \Lambda_{i, i'}] - \cN(\bzero, I_{k^2})\{\Lambda_{i, i'}\}\right| \\
&\leq \left|\Pr[H \in \Lambda_{i, i'}] - \cN(\bzero, \Sigma)\{\Lambda_{i, i'}\}\right| \\
&\qquad+ \left|\cN(\bzero, \Sigma)\{\Lambda_{i, i'}\} - \cN(\bzero, I_{k^2})\{\Lambda_{i, i'}\}\right|
&\overset{\eqref{eq:mult-berry-esseen-app}, \eqref{eq:tv-gaussian-spherical}}{\leq} \gamma.
\end{align*}
Observe that by definition of $\Lambda_{i, i'}$ and non-atomicity of $\cD$, it holds that $p_{(i, i')} = \Pr[H \in \Lambda_{i, i'}]$. 
Moreover, due to symmetry, we have
\begin{align*}
\cN(\bzero, I_{k^2})\{\Lambda_{i, i'}\} = \cN(\bzero, I_{k^2})\{\Lambda_{j, j'}\}
\end{align*}
for all $j, j' \in [k]$. 
Note that $\bigcup_{j,j' \in [k]} \Lambda_{j, j'} = \R^{k^2}$ and the intersection of any two distinct sets $\Lambda_{j, j'}$ has measure zero (with respect to $\cN(\bzero, I_{k^2})$). 
As such, we have that
$\cN(\bzero, I_{k^2})\{\Lambda_{i, i'}\} = \frac{1}{k^2}.$
Combining this with the inequality above concludes our proof.
\end{proof}

\subsection{Proof of \Cref{lem:maxprob-main}}

We introduce an additional notation.
For any random variable~$X$, let $\maxprob(X) = \max_{x \in \supp(X)} \Pr[X = x]$. 
Sometimes we also write the distribution of $X$ instead of $X$ itself. 
We use the following properties of $\maxprob$, whose proofs we defer to \Cref{subsec:properties-of-maxprob}.
\begin{lemma} \label{lem:max-prob-conv}
For independent $X, Y \in \Z^d$, it holds that $\maxprob(X + Y) \leq \maxprob(X)$.
\end{lemma}

\begin{lemma} \label{lem:multi-max-prob}
For any positive integer $r \geq k$, the following inequality holds:
$\maxprob\left(\multi{r}{\frac{1}{k} \cdot \bone}\right) \leq \frac{k^{k/2}}{(2\pi r)^{\frac{k-1}{2}}}$.
\end{lemma}

We also need Hoeffding's inequality, which is stated in terms of the binomial distribution for convenience.

\begin{lemma}[Hoeffding's inequality] \label{lem:chernoff}
For $\psi \in (0, 1)$, $t > 0$, and $n \in \N$,
$\Pr_{S \sim \bino{n}{\psi}}[S < n \psi - t] \leq \exp\left(-2t^2/n\right)$.
\end{lemma}
 
We wish to upper bound $\Pr\left[\tX^1 + \cdots + \tX^n = \frac{n}{k}\cdot\bone\right]$. By definition, this is at most $\maxprob\left(\tX^1 + \cdots + \tX^n\right)$. 
To bound this quantity, we reinterpret $\tX^a$ so that it is a mixture distribution between the uniform distribution on $\{\be_1, \dots, \be_k\}$ and a ``leftover'' distribution. 
In other words, for each $\tX^a$, we can toss a (biased) coin and, based on the outcome, sample $\tX^a$ either from the uniform distribution or from the leftover distribution.
Using \Cref{lem:chernoff}, we show that $\tX^a$ is drawn from the uniform distribution for most $a$. 
For these $\tX^a$, their sum exactly follows the multinomial distribution, so we can apply \Cref{lem:multi-max-prob}. 
Finally, \Cref{lem:max-prob-conv} allows us to handle the leftover.

\begin{proof}[Proof of \Cref{lem:maxprob-main}]
It suffices to bound $\maxprob(\tX^1 + \cdots + \tX^n)$.

Instead of sampling each $\tX^a$ directly,
we can alternatively sample $\tX^a$ as follows.
\begin{itemize}
\item Sample $Z^a \sim \ber{1 - \zeta}$.
\item Sample $U^a$ uniformly at random from $\{\be_1, \dots, \be_k\}$.
\item Sample $V^a$ from $\{\be_1, \dots, \be_k\}$ where $\Pr[V^a = \be_i] = \frac{1}{\zeta}\left(\Pr[\tX^a = \be_i] - (1 - \zeta)\frac{1}{k}\right)$.
\item If $Z^a = 1$, let $\tX^a = U^a$; if $Z^a = 0$, let $\tX^a = V^a$.
\end{itemize}
Observe that $V^a$ is a valid probability distribution. 
Also, one can check that the process defined above results in the correct probability distribution for $\tX^a$.

Therefore, $\maxprob(\tX^1 + \cdots + \tX^n)$ can be written as
\begin{align*}
&\max_{\bx \in \Delta_n^{k-1}} \Pr[\tX^1 + \cdots + \tX^n = \bx] \\
&= \max_{\bx \in \Delta_n^{k-1}} \sum_{S \subseteq N} \Pr\left[\sum_{a \in S} U^a + \sum_{a' \notin S} V^{a'} = \bx\right] \\ &\qquad \qquad \qquad \cdot \Pr\left[S = \{a \in N \mid Z^a = 1\} \right] \\
&\leq \sum_{S \subseteq N} \maxprob\left(\sum_{a \in S} U^a + \sum_{a' \notin S} V^{a'}\right) \\ &\qquad \qquad \qquad 
\cdot \Pr\left[S = \{a \in N \mid Z^a = 1\} \right] \\
&\overset{\text{(\Cref{lem:max-prob-conv})}}{\leq}  \sum_{S \subseteq N} \maxprob\left(\sum_{a \in S} U^a\right) \\ &\qquad \qquad \qquad  \cdot \Pr\left[S = \{a \in N \mid Z^a = 1\} \right].
\end{align*}
Observe that $\sum_{a \in S} U^a$ is distributed as $\multi{|S|}{\frac{1}{k} \cdot \bone}$ and that $|S| = |\{a \in N \mid Z^a = 1\}|$ is distributed as $\bino{n}{1 - \zeta}$. 
Thus, we can write the above inequality as
\begin{align*}
&\maxprob(\tX^1 + \cdots + \tX^n) \\
&\leq \E_{r \sim \bino{n}{1 - \zeta}}\left[\maxprob\left(\multi{r}{\frac{1}{k}\cdot\bone}\right)\right] \\
&\leq \Pr_{r \sim \bino{n}{1 - \zeta}}[r \leq (1 - 2\zeta)n] \\& + \E_{r \sim \bino{n}{1 - \zeta}}\left[\maxprob\left(\multi{r}{\frac{1}{k}\cdot\bone}\right) \middle\vert r > (1 - 2\zeta)n\right],
\end{align*}
where the last inequality follows from the law of total expectation.
Applying Lemmas~\ref{lem:chernoff} and \ref{lem:multi-max-prob} respectively to bound the two remaining terms completes the proof.
\end{proof}

\subsection{Properties of $\maxprob$: Proofs of \Cref{lem:max-prob-conv,lem:multi-max-prob}}
\label{subsec:properties-of-maxprob}

\begin{proof}[Proof of \Cref{lem:max-prob-conv}]
For any $z \in \supp(X + Y)$, we have
\begin{align*}
\Pr[X + Y = z] &= \sum_{y \in \supp(Y)} \Pr[X = z - y] \Pr[Y = y] \\ &\leq \sum_{y \in \supp(Y)} \maxprob(X) \Pr[Y = y] = \maxprob(X).
\end{align*}
This implies that $\maxprob(X + Y) \leq \maxprob(X)$.
\end{proof}

For \Cref{lem:multi-max-prob}, we will use the following properties of the gamma function. 
These properties are well-known and their proofs can be found in standard textbooks on the topic (e.g., \citealt{Artin64}).
Recall that a function is said to be logarithmically convex if its logarithm is convex.
\begin{lemma} \label{lem:gamma-property}
$\Gamma: \R \to \R$ is a function that satisfies
\begin{enumerate}[(i)]
\item $\Gamma(1 + r) = r!$ for all $r \in \Z_{\geq 0}$;
\item $\Gamma$ is logarithmically convex on $(0, \infty)$;
\item For $x \geq 1$, $\Gamma(1 + x) = \sqrt{2\pi x}(x/e)^x e^{\mu(x)}$, where $\mu: [1, \infty) \to \R_{\ge 0}$ is decreasing.
\end{enumerate}
\end{lemma}

\begin{proof}[Proof of \Cref{lem:multi-max-prob}]
Consider any $\bx \in \Delta_r^{k-1}$, we have
\begin{align*}
\multi{r}{\frac{1}{k} \cdot \bone}\{\bx\} &= \frac{r!}{\prod_{i \in [k]} x_i!} \cdot \frac{1}{k^r} \\
&= \frac{\Gamma(r+1)}{\prod_{i \in [k]} \Gamma(x_i+1)} \cdot \frac{1}{k^r} \\
&\leq \frac{\Gamma(r+1)}{\Gamma(r/k+1)^k} \cdot \frac{1}{k^r} \\
&\leq \frac{\sqrt{2 \pi r}(r/e)^r}{\left(\sqrt{2 \pi r/k}\left(\frac{r}{ke}\right)^{r/k}\right)^k} \cdot \frac{1}{k^r} \\
&= \frac{k^{k/2}}{(2\pi r)^{\frac{k-1}{2}}},
\end{align*}
where we use the properties of the gamma function from \Cref{lem:gamma-property} in the respective order.
\end{proof} 

We remark that bounds in the same vein as \Cref{lem:maxprob-main} have been shown, e.g., by \citet{Ushakov86} and \citet{PostnikovYu88}. 
However, these are not sufficient for our purposes, since either they require the distributions to be identical or the provided probability bound is not tight for $k > 2$ (up to a constant factor). 

\subsection{Proof of \Cref{lem:prob-smallest-group}}

The proof of \Cref{lem:prob-smallest-group} requires a few intermediate lemmas. 
To state these lemmas, we need an additional notation. 
For any non-atomic real-valued distribution $\cP$ and any $p \in (0, 1)$, let $Q_\cP(p)$ denote its $p$-th quantile, i.e., the (unique) point $x$ such that $\Pr_{X \sim \cP}[X < x] = p$. 

We start with the following inequality, which is an immediate consequence of Chebyshev's inequality.
\begin{lemma} \label{lem:quantile-bound}
For any non-atomic real-valued distribution $\cP$ with mean $\mu$ and standard deviation $\sigma_{\cP}$ and any $p \in (0, 1)$,
\begin{align*}
\mu - \frac{\sigma_{\cP}}{\sqrt{p}} \leq Q_\cP(p) \leq \mu + \frac{\sigma_{\cP}}{\sqrt{1- p}}.
\end{align*}
\end{lemma}

Next is a bound on the variance of a distribution after conditioning it to be larger than a certain value.
\begin{lemma} \label{lem:cond-var}
For any non-atomic real-valued distribution $\cP$ with standard deviation $\sigma_{\cP}$ and any $p \in (0, 1)$, let $\cP_{\geq x}$ denote the conditional distribution of $\cP$ on $[x, \infty)$. 
Then, we have $\var(\cP_{\geq Q_\cP{(p)}}) \leq \sigma_{\cP}^2 / (1 - p)$.
\end{lemma}
\begin{proof}
For any distribution~$\cY$ with mean~$\mu_\cY$ and any real number~$c$, one can check that $\E_{Y\sim \cY}[(Y-c)^2] = \var(\cY) + (c-\mu_Y)^2 \ge \var(\cY)$.

Let $\mu_\cP$ denote the mean of $\cP$. 
We have
\begin{align*}
&\var(\cP_{\geq Q_\cP(p)}) \\
&\leq \E_{X \sim \cP_{\geq Q_\cP(p)}}[(X - \mu_\cP)^2] \\
&= \frac{\E_{X \sim \cP}[(X - \mu_\cP)^2 \cdot \bone[X \geq Q_{\cP}(p)]]}{1 - p} \\
&\leq \frac{\E_{X \sim \cP}[(X - \mu_\cP)^2]}{1 - p} \\
&= \frac{\sigma_{\cP}^2}{1 - p},
\end{align*}
as desired.
\end{proof}

We next bound the probability that two independent and identically distributed (i.i.d.) random variables differ by at most a certain amount.

\begin{lemma} \label{lem:bounded-diff}
For any non-atomic real-valued distribution $\cP'$ with standard deviation $\sigma'$ and any $\xi > 0$, we have
\begin{align*}
\Pr_{Y_1, Y_2 \overset{\iid}{\sim} \cP'}[|Y_1 - Y_2| < \xi] \geq \min\left\{\frac{1}{8}, \frac{\xi}{32\sigma'}\right\}.
\end{align*}
\end{lemma}

\begin{proof}
We may assume that $\xi \leq 4\sigma'$, as the right-hand side does not increase further if $\xi > 4\sigma'$.
\Cref{lem:quantile-bound} implies that $|Q_{\cP'}(0.25) - Q_{\cP'}(0.75)| \le 4\sigma'$. 
Let $T = \lceil 4\sigma'/\xi \rceil \ge 1$ and $I_1, \dots, I_T$ be any partition of  $[Q_{\cP'}(0.25), Q_{\cP'}(0.75)]$ into disjoint intervals of length at most $\xi$ each. 
We have
\begin{align*}
&\Pr_{Y_1, Y_2 \overset{\iid}{\sim} \cP'}[|Y_1 - Y_2| < \xi] \\
&\geq \sum_{t \in [T]} \Pr_{Y_1, Y_2 \overset{\iid}{\sim} \cP'}[Y_1, Y_2 \in I_t] \\
&= \sum_{t \in [T]} \Pr_{Y \sim \cP'}[Y \in I_t]^2 \\
&\overset{(\star)}{\geq} \frac{1}{T} \left(\sum_{t \in [T]} \Pr_{Y \sim \cP'}[Y \in I_t]\right)^2 \\
&= \frac{1}{T} \cdot \Pr_{Y \sim \cP'}[Y \in [Q_{\cP'}(0.25), Q_{\cP'}(0.75)]]^2 \\
&= \frac{1}{4T} \\
&\geq \frac{\xi}{32\sigma'},
\end{align*}
where $(\star)$ follows from the Cauchy--Schwarz inequality and the last inequality from our choice of $T$.
\end{proof}

Using the two bounds above, we now derive an inequality concerning the difference between the maximum and a non-maximum among $k$ i.i.d.~random variables. 

\begin{lemma} \label{lem:diff-main-lb}
Let $\cP$ be any non-atomic real-valued distribution with standard deviation $\sigma_{\cP}$.
Let $Z_1, Z_2, \dots, Z_k \overset{\iid}{\sim} \cP$. Then, for any $\xi > 0$,
\begin{align*}
&\Pr[Z_1 > Z_2, \dots, Z_k \text{ and } Z_1 - Z_k \leq \xi] \\
&\geq \frac{1}{16k^2} \cdot \min\left\{\frac{1}{8}, \frac{\xi}{64\sigma_{\cP}\sqrt{k}}\right\}.
\end{align*}
\end{lemma}

\begin{proof}
By symmetry, we have
\begin{align*}
&\Pr[Z_1 > Z_2, \dots, Z_k \text{ and } Z_1 - Z_k \leq \xi] \\
&= \frac{1}{2}\big(\Pr[Z_1 > Z_2, \dots, Z_k \text{ and } Z_1 - Z_k \leq \xi] \\
&\qquad + \Pr[Z_k > Z_1, \dots, Z_{k-1} \text{ and } Z_k - Z_1 \leq \xi]\big) \\
&\geq \frac{1}{2} \cdot \Pr[Z_1, Z_k > Z_2, \dots, Z_{k - 1} \text{ and } |Z_1 - Z_k| \leq \xi].
\end{align*}
To bound this term further, let $p = 1 - \frac{1}{2k}$ and $q = Q_{\cP}(p)$. We have
\begin{align*}
&\Pr[Z_1, Z_k > Z_2, \dots, Z_{k - 1} \text{ and } |Z_1 - Z_k| \leq \xi] \\
&\geq \Pr[Z_1, Z_k \geq q > Z_2, \dots, Z_{k - 1} \text{ and } |Z_1 - Z_k| \leq \xi] \\
&= (1 - p)^2 p^{k - 2} \cdot \Pr[|Z_1 - Z_k| \leq \xi \mid Z_1, Z_k \geq q] \\
&= (1 - p)^2 p^{k - 2} \cdot \Pr_{Z'_1, Z'_k \overset{\iid}{\sim} \cP_{\geq q}}[|Z'_1 - Z'_k| \leq \xi] \\
&\geq (1 - p)^2 p^{k - 2} \cdot \min\left\{\frac{1}{8}, \frac{\xi\sqrt{1 - p}}{32\sigma_{\cP}}\right\},
\end{align*}
where the last inequality follows from \Cref{lem:cond-var,lem:bounded-diff}.

Finally, combining the two inequalities above and using the definition of $p$, we have
\begin{align*}
&\Pr[Z_1 > Z_2, \dots, Z_k \text{ and } Z_1 - Z_k \leq \xi] \\
&\geq \frac{1}{2} \cdot \frac{1}{4k^2} \cdot \left(1 - \frac{1}{2k}\right)^{k - 2} \cdot \min\left\{\frac{1}{8}, \frac{\xi}{64\sigma_{\cP}\sqrt{k}}\right\} \\
&\geq \frac{1}{16k^2} \cdot \min\left\{\frac{1}{8}, \frac{\xi}{64\sigma_{\cP}\sqrt{k}}\right\},
\end{align*}
where the last inequality is due to Bernoulli's inequality.
\end{proof}

We are now ready to prove \Cref{lem:prob-smallest-group}. 
At a high level, we first recall that in the case where $A$ is balanced, each $p_i^{A}$ is equal to $1/k$. 
We then analyze the ``probability deficit'' that $p_i^{A}$ incurs when a good is added to some other bundle $A_{i'}$. 
\Cref{lem:diff-main-lb} ensures that, before the addition, the difference between $u_a(A_i)$ and $u_a(A_{i'})$ is small. 
This allows us to lower-bound the probability deficit, which is incurred when the additional value exceeds the difference.

\begin{proof}[Proof of \Cref{lem:prob-smallest-group}]
Assume without loss of generality that $|A_1| \leq |A_2| \leq \dots \leq |A_k|$. 
We will show that the desired inequality holds for $i = 1$. 
The case $|A_1| = 0$ is trivial, so we henceforth assume $|A_1| \geq 1$.
Let $A'_2, \dots, A'_k$ be arbitrary subsets of $A_2, \dots, A_k$ of size $|A_1|$, respectively, and let $Z_1 = u_a(A_1), Z_2 = u_a(A'_2), \dots, Z_k = u_a(A'_k)$. 
Since $Z_1,\dots,Z_k$ are i.i.d.~and non-atomic, we have
\begin{align*}
\Pr[Z_1 \geq Z_2, \dots, Z_k] = \frac{1}{k}.
\end{align*}
Thus, we can bound the desired probability as
\begin{align*}
&\Pr[u_a(A_1) \geq u_a(A_2), \dots, u_a(A_k)] \\
&\leq \Pr[u_a(A_1) \geq u_a(A'_2), \dots, u_a(A'_{k - 1}), u_a(A_k)] \\
&= \Pr[Z_1 \geq Z_2, \dots, Z_k \text{ and } u_a(A_k) \le Z_1] \\
&= \Pr[Z_1 \geq Z_2, \dots, Z_k] \\&\qquad - \Pr[Z_1 \geq Z_2, \dots, Z_k \text{ and } u_a(A_k) > Z_1] \\
&= \frac{1}{k} - \Pr[Z_1 \geq Z_2, \dots, Z_k \text{ and } u_a(A_k) > Z_1].
\end{align*}

Now, observe that $u_a(A_k) = Z_k + u_a(A_k \setminus A'_k)$. 
Since $m$ is not divisible by $k$, we have $A_k \setminus A'_k \ne \emptyset$. 
By definition of quantiles, it holds that $\Pr[u_a(A_k \setminus A'_k) > Q_{\cD}(1/2)] \geq 1/2$. 
Moreover, $u_a(A_k \setminus A'_k)$ is independent of $Z_1, \dots, Z_k$. 
Hence, we have
\begin{align*}
&\Pr[Z_1 \geq Z_2, \dots, Z_k \text{ and } u_a(A_k) > Z_1] \\
&= \Pr[Z_1 \geq Z_2, \dots, Z_k \text{ and } Z_1 - Z_k < u_a(A_k\setminus A_k')] \\
&\geq \Pr[Z_1 \geq Z_2, \dots, Z_k \text{ and }  \\
&\qquad\qquad Z_1 - Z_k < Q_{\cD}(1/2) < u_a(A_k\setminus A_k')] \\
&\geq \frac{1}{2} \cdot \Pr[Z_1 \geq Z_2, \dots, Z_k \text{ and } Z_1 - Z_k < Q_{\cD}(1/2)].
\end{align*}

For any random variable $Y$ supported on $[0,1]$, it holds that 
\[
\var(Y) = \E[Y^2] - \E[Y]^2 \le \E[Y] - \E[Y]^2 \le 1/4. 
\]
Since each of $Z_1,\dots,Z_k$ is a sum of $|A_1|$ independent random variables supported on $[0,1]$, we have 
\[
\var(Z_1), \dots, \var(Z_k) \leq \frac{|A_1|}{4} \leq \frac{m}{4k}.
\]
Moreover, $Q_{\cD}(1/2) > 0$ since $\cD$ is non-atomic and supported on $[0, 1]$. 
Thus, we may apply \Cref{lem:diff-main-lb} to get
\begin{align*}
&\Pr[Z_1 \geq Z_2, \dots, Z_k \text{ and } Z_1 - Z_k < Q_{\cD}(1/2)] \\
&\geq \frac{Q_{\cD}(1/2)}{512k^2\sqrt{m}}.
\end{align*}
Finally, putting everything together yields the desired bound for $\alpha = \frac{Q_{\cD}(1/2)}{1024}$. 
\end{proof}

\subsection{Proof of \Cref{lem:prob-single-alloc-ef}}

From \Cref{obs:event-prob-simplified}, \Cref{cor:lem:stirling-mult-equip}, and the fact that $n \geq k$, we have $\Pr[\cE_A] \leq \frac{\sqrt{k}}{e^{n \cdot \kl{\frac{\bone}{k}}{\bp^{A}}}}$. 
Also, Lemmas~\ref{lem:pinsker} and~\ref{lem:prob-smallest-group} respectively imply that $\kl{\frac{\bone}{k}}{\bp^{A}} \geq 2 \cdot \tv{\frac{\bone}{k}}{\bp^{A}}^2 \geq \frac{\alpha^2}{k^4 m}$. Combining these yields the claimed bound.

\end{document}